\documentclass{report}
\usepackage{amsfonts}
\usepackage{amsmath}
\usepackage{graphicx}

\newtheorem{theorem}{Theorem}

\newtheorem{lemma}[theorem]{Lemma}

\newtheorem{proposition}[theorem]{Proposition}
\newtheorem{remark}[theorem]{Remark}

\newenvironment{proof}[1][Proof]{\noindent\textbf{#1.} }{\ \rule{0.5em}{0.5em}}
\begin{document}

\title{Importance Sampling for rare events and conditioned random walks}
\author{M. Broniatowski$^{(1)}$, Y.Ritov$^{(2)}$ \\
$^{(1)}$LSTA, Universit\'{e} Paris 6, France\\
michel.broniatowski@upmc.fr\\
$^{(2)}$Dpt of Statistics, Hebrew University, Jerusalem Isra\"{e}l\\
\qquad yaacov.ritov@gmail.com}
\date{October 2009}
\maketitle

\begin{abstract}
This paper introduces a new Importance Sampling scheme, called Adaptive
Twisted Importance Sampling, which is adequate for the improved estimation
of rare event probabilities in he range of moderate deviations pertaining to
the empirical mean of real i.i.d. summands. It is based on a sharp
approximation of the density of long runs extracted from a random walk
conditioned on its end value.
\end{abstract}

\tableofcontents

\section{Introduction and notation}

Importance Sampling procedures aim at reducing the calculation time which is
necessary in order to evaluate integrals, often in large dimension. We
consider the case when the integral to be numerically computed is the
probability of an event defined by a large number of random components; this
case has received quite a lot of attention, above all when the event is of
\textit{small} probability, typically of order $10^{-8}$ or so, as occurs
frequently in industrial applications or in communication devices. The order
of magnitude of the probability to be estimated is here somehow larger, and
aims at coping with "moderate probabilities" as dealt with in statistics.
The basic situation in IS can be stated as follows.

Let $\mathbf{Z}$ be some random variable, say on $\mathbb{R},$ with
probability measure $P$ and density $p$. Let $A$ be a subset of $\mathbb{R}$
with $P(A)>0.$ Let $Z_{1}^{L}:=(Z_{1},...,Z_{L})$ denote a sample of i.i.d.
observations of $\mathbf{Z}$ . By the law of large numbers
\begin{equation}
P_{L}:=\frac{1}{L}\sum_{l=1}^{L}\mathbf{1}_{A}(Z_{l})  \label{LLN}
\end{equation}%
estimates $P(A)$ without bias, when the $Z_{i}^{\prime }s$ are sampled under
the density $p.$ An altenative unbiased estimate of $P(A)$ can be defined
through
\begin{equation}
P_{L}^{g}:=\frac{1}{L}\sum_{l=1}^{L}\frac{p(Y_{l})}{g(Y_{l})}\mathbf{1}%
_{A}(Y_{l})  \label{ISg}
\end{equation}%
for all density $g$ when the support of $p$ is a subset of the support of $g$%
, and the $Y_{i}^{\prime }s.$are i.i.d. observations of a r.v. $\mathbf{Y}$
with density $g.$ As is well known the optimal choice for the IS sampling
density $g$ is $p_{\mathbf{Z}/A}$ , the density of $\mathbf{Z}$ conditioned
upon the event $(\mathbf{Z}\in A),$ unfortunately an unpracticable choice
which presumes the knowledge of $P(A),$ the quantity to be estimated. Would
this sampling density be at hand, the required number $L$ of replications of
$\mathbf{Y}$ to be performed would reduce to $1$ and the estimate would be
exactly $P(A).$ This fact motivates efforts in order to approximate $p_{%
\mathbf{Z}/A}$ in the case when the variable $\mathbf{Z}$ has a distribution
which allows it. Sometimes the random variable $\mathbf{Z}$ is obtained as a
function of a large number of random variables, say $\mathbf{X}%
_{1}^{n}:=\left( \mathbf{X}_{1},...,\mathbf{X}_{n}\right) $ and the event $(%
\mathbf{Z}\in A)$ is of small or moderate probability. Also the density of $%
\mathbf{Z}$ cannot be evaluated analytically, due to the very definition of $%
\mathbf{Z}$, but the random variables $\mathbf{X}_{i}$ 's have known
distribution. This happens for instance when $\mathbf{Z}$ is a moment
estimator or when it is the linear part of the expansion of an M or L
-estimate (see Section 4).\ The example which we have in mind is the
following, which helps as a benchmark case in the IS literature.\

The r.v's $\mathbf{X}_{i}^{\prime }s$ are i.i.d. , are centered with
variance 1, with common density $p_{\mathbf{X}}$ on $\mathbb{R}$, and
\begin{equation*}
\mathbf{Z}:=\frac{1}{n}\sum_{i=1}^{n}\mathbf{X}_{i}=:\frac{1}{n}\mathbf{S}%
_{1}^{n}
\end{equation*}%
is the empirical mean of the $\mathbf{X}_{i}^{\prime }s.$ The set $A$ is
\begin{equation}
A:=(a_{n},\infty )  \label{A}
\end{equation}
where $a_{n}$ tends slowly to $E(\mathbf{X}_{1})$ from above and we intend
to estimate%
\begin{equation*}
P_{n}:=P\left( \frac{1}{n}\mathbf{S}_{1}^{n}\in A\right)
\end{equation*}%
for large but fixed $n.$ Many asymptotic results provide sharp estimates for
$P(\mathbf{Z}\in A)$ but it is a known fact that asymptotic expansions are
not always good tools when dealing with numerical approximations for fixed
(even large) $n.$ For example, citing Ermakov (2004, p 624, \cite%
{Ermakov2003}), the Berry-Esseen approximation for the evaluation of risks
of order $10^{-2}$ in testing is pertinent for sample sizes of order
5000-10000; also the accuracy of available moderate deviation probabilities
as developped by Inglot, Kallenberg and Ledwina in \cite%
{InglotKallenbergLedwina1992} has not been investigated. This motivates our
interest in numerical techniques in this field.

According to (\ref{LLN}) the basic estimate of $P(\mathbf{Z}\in A)$ is
defined as follows: generate $L$ i.i.d. samples $X_{1}^{n}(l)$ with
underlying density $p_{X}$ and define
\begin{equation*}
P^{(n)}(\mathcal{E}_{n}):=\frac{1}{L}\sum_{l=1}^{L}\mathbf{1}_{\mathcal{E}%
_{n}}\left( X_{1}^{n}(l)\right)
\end{equation*}%
where
\begin{equation}
\mathcal{E}_{n}:=\left\{ (x_{1},...,x_{n})\in \mathbb{R}%
^{n}:s_{1}^{n}/n>a_{n}\right\} .  \label{E_n}
\end{equation}%
Here $s_{1}^{n}:=x_{1}+...+x_{n}$ $.$ The statistics $P^{(n)}(\mathcal{E}%
_{n})$ estimates the \textit{moderate deviation} probability of the sample
mean of the $\mathbf{X}_{i}^{\prime }s.$ Also denoting $g$ a sampling
density of the vector $Y_{1}^{n}$ the associated IS estimate is
\begin{equation}
P_{g}^{(n)}(\mathcal{E}):=\frac{1}{L}\sum_{l=1}^{L}\frac{p_{X}\left(
Y_{1}^{n}(l)\right) }{g\left( Y_{1}^{n}(l)\right) }\mathbf{1}_{\mathcal{E}%
}\left( Y_{1}^{n}(l)\right) .  \label{FORM IS}
\end{equation}

In the range of moderate deviations the two major contributions to IS
schemes for the estimation of $P_{n}$ are Fuh and Hu \cite{FuhWu2004} and
Ermakov \cite{Ermakov2007}. The paper by Fuh and Hu does not consider events
of moderate deviations as intended here; it focuses on IS schemes for the
estimation of $P(Z\in A)$ where $Z$ is a given multinormal random vector and
$A$ is a fixed set in $\mathbb{R}^{d}.$ The authors consider efficiency with
respect to the variance of the estimate and state that for the case of
interest the efficient sampling scheme is deduced from the distribution of $%
Z $ by a shift in the mean inside the set $A.$ The papers by Ermakov instead
handle similar problems as we do. Ermakov's 2007 paper\ \cite{Ermakov2007}
considers a sampling scheme where $g$ is the density of i.i.d. components.
He proves that this scheme is efficient in the sense that the computational
burden necessary to obtain a relative precision of the estimate with respect
to $P_{n}$ does not grow exponentially as a function of $n.$ He considers
statistics of greater generality than the sample mean, such as M and L
estimators; in the range of moderate deviations the asymptotic behavior of
those objects is captured however through their linear part which is the
empirical mean of their influence function, which puts the basic situation
back at the center of the scene. We discuss efficiency in Section 3 and
present some results in connection with Ermakov's pertaining to M and L
estimators in Section 4.

\bigskip

\bigskip

The numerator in the expression \ (\ref{FORM IS}) is the product of the $p_{%
\mathbf{X}_{1}}(Y_{i})^{\prime }s$ while the denominator need not be a
density of i.i.d. copies evaluated on the $Y_{i}^{\prime }s$. Indeed the
optimal choice for $g$ is the density of $\mathbf{X}_{1}^{n}$ conditioned
upon $\mathcal{E}_{n},$ say $p_{\mathbf{X}_{1}^{n}/\mathcal{E}_{n}}.$

Since the optimal solution is known to be $p_{\mathbf{X}_{1}^{n}/\mathcal{E}%
_{n}}$, the best its approximation, the best the sampling scheme, at least
when it does not impose a large calculation burden; classical sampling
schemes consist in simulation of independent copies of r.v.'s $Y_{i}(l)$ ,$%
1\leq i\leq n$, and efficiency is defined in terms of variance of the
estimate inside this class of sampling, which, by nature, is suboptimal with
respect to sampling under good approximations of $p_{\mathbf{X}_{1}^{k}/%
\mathcal{E}_{n}}$ for long runs, i.e. for large $k=k_{n}.$The present paper
explores the choice of good sampling schemes from this standpoint. Obviously
mimicking the optimal scheme \ \ results in a net gain on the number $L$ of
replications of the runs which are necessary to obtain a given accuracy of
the estimate with respect to $P_{n}$ $.$ However the criterion which we
consider is different from the variance, and results as an evaluation of the
MSE of our estimate on specific subsets of the runs generated by the
sampling scheme, which we call typical subsets, namely having probability
going to $1$ under the sampling scheme as $n$ increases. On such sets, the
MSE is proved to be of very small order with respect to the variance of the
classical estimate, whose MSE\ cannot be diminuished on any such typical
subsets. We believe that this definition makes sense and prove it also
numerically. This is the scope of Section 3 in which \ it will be shown that
the relative gain in terms of simulation runs necessary to perform an $%
\alpha \%$ relative error on $P_{n}$ drops by a factor $\sqrt{n-k}/\sqrt{n}$
with respect to the classical IS\ scheme.

Our proposal therefore hinges on the local approximation of the conditional
distribution of longs runs $\mathbf{X}_{1}^{k}$ from $\ \mathbf{X}_{1}^{n}.$
This cannot be achieved through the classical theory of moderate deviations,
first developped by De Acosta and more recently by Ermakov; at the contrary
the ad hoc procedure developped in the range of large deviations by Diaconis
and Freedman \cite{DiaconisFreedman1988}\ for the local approximation of the
conditional distribution of $\mathbf{X}_{1}^{k}$ given the value of $\mathbf{%
S}_{1}^{n}$ is the starting point of the present approach. We find it useful
to briefly expose these two different points of view. We also mention the
approximation technique for moderate deviations of sub linear functionals of
the empirical measure by Inglot, Kallenberg and Ledwina \cite%
{InglotKallenbergLedwina1992}, based on strong approximation techniques;
these results provide explicit equivalents for the probability of moderate
deviations, but do not lead to adequate approximations for the obtention of
their numerical counterparts by IS methods.

The following notation and assumptions will be kept throughout this paper.

We assume that $\mathbf{X}_{1}$ satisfies the Cramer condition, i.e. $%
\mathbf{X}_{1}$ has a finite moment generating function $\Phi (t):=E\exp t%
\mathbf{X}_{1}$ in a non void neighborhood of $0;$ denote
\begin{equation}
m(t):=\frac{d}{dt}\log \Phi (t)  \label{mean tilted}
\end{equation}%
and
\begin{equation}
s^{2}(t):=\frac{d}{dt}m(t)  \label{var tilted}
\end{equation}%
when defined. The values of $m(t^{\alpha }):=\frac{d}{dt}\log \Phi
(t^{\alpha })$ and $s^{2}(t^{\alpha }):=\frac{d}{dt}m(t^{\alpha })$ are the
expectation and the variance of the \textit{tilted} density
\begin{equation}
\pi ^{\alpha }(x):=\frac{\exp t^{\alpha }x}{\Phi (t^{\alpha })}p(x)
\label{tilted density}
\end{equation}%
where $t^{\alpha }$ is the only solution of the equation $m(t)=\alpha $ when
$\alpha $ belongs to the support of $p.$ Denote $\Pi ^{\alpha }$ the
probability measure with density $\pi ^{\alpha }$. The \textit{Chernoff
function} of $\mathbf{X}_{1}$ is%
\begin{equation*}
I(x):=\sup_{t}tx-\log \Phi (t)
\end{equation*}%
for $x$ in the support of $\mathbf{X}_{1}$ and it holds%
\begin{eqnarray*}
\frac{d}{dx}I(x) &=&m^{\leftarrow }(x) \\
\frac{d^{2}}{dx^{2}}I(x) &=&\frac{1}{s^{2}\circ m^{\leftarrow }(x)}
\end{eqnarray*}%
where $m^{\leftarrow }(x)$ denotes the reciprocal function of $m.$

Denote
\begin{equation*}
\varphi (s):=\int_{-\infty }^{+\infty }e^{isx}p_{\mathbf{X}}(x)dx
\end{equation*}%
the characteristic function of $\mathbf{X}_{1}\mathbf{.}$ Assume that
\begin{equation}
\int_{-\infty }^{+\infty }\left\vert \varphi (s)\right\vert ^{\nu }ds<\infty
\label{f.c. Edgeworth}
\end{equation}%
for some $\nu \geq 1.$ This condition entails the validity of the Edgeworth
expansions to be used in the sequel (see e.g. Feller \cite{Feller1971}).

The notation $p(\mathbf{X}=x)$ is used to denote the value of the density $p$
of the r.v. $\mathbf{X}$ at point $x.$ The notation $p(\mathbf{S}_{1}^{n}=s)$
is used to define the value of the density of the r.v. $\mathbf{S}_{1}^{n}$
under $p$, i.e. when the summands are i.i.d. with density $p.$ Also we may
write $p\left( f\left( \mathbf{X}_{1}^{n}\right) =u\right) $ to denote the
density (on the corresponding image space) of some function $f$ of the
sample $\mathbf{X}_{1}^{n}.$ We write $\mathfrak{P}_{n}$ the distribution of
$\mathbf{X}_{1}^{n}$ given $\mathcal{E}_{n}$ and $\mathfrak{p}_{n}$ its
density. The symbol $\mathfrak{n}$ denotes the standard normal density on $%
\mathbb{R}.$

\subsection{From moderate deviations to conditional distributions}

A basic requirement for a good IS sampling scheme is that it mimicks the
conditional \textit{density }$p_{\mathbf{X}_{1}^{n}/\mathcal{E}_{n}}.$ We
first expose a general argument in this direction in order to clarify that
there is no bypass through the general theory of large or moderate
deviations to achieve this goal. Also the present discussion motivates the
choices of classical IS sampling schemes (Ermakov), emphasizing that the
general theory provides the proof that the \textit{marginal }conditional
distribution of $\mathbf{X}_{1}^{n}$ under $\mathcal{E}_{n}$ is well
approximated by $\Pi ^{a_{n}}$ a statement which is usually refered to as a
\textit{Gibbs conditional principle.} We need some tools from the moderate
deviation principle as developped by \cite{Ermakov2003} following \cite%
{deAcosta1992}.

Let $F$ be a class of measurable functions defined on $\mathbb{R}$ and $%
M_{F} $ be the class of all signed finite measures on $\mathbb{R}$ which
satisfy%
\begin{equation*}
\int \left\vert f\right\vert d\left\vert Q\right\vert <\infty \text{ for all
}f\text{ in }F.
\end{equation*}%
On $M_{F}$ define the $\tau _{F}$ topology, which is the coarsest for which
all mappings $f\rightarrow \int fdQ$ ($Q\in M_{F}$) are continuous for all $%
f $ in $F$. For $P$ a probability measure and $Q$ in $M(\mathbb{R)}$ the
so-called Chi-square distance between $P$ and $Q$ is defined through%
\begin{equation*}
\chi ^{2}(Q,P):=\frac{1}{2}\int \left( \frac{dQ}{dP}-1\right) ^{2}dP
\end{equation*}%
whenever $Q$ is absolutely continuous with respect to $P,$ and equals $%
+\infty $ otherwise.\

$\ $The following \textit{moderate deviation} Sanov result holds; see \cite%
{Ermakov2007}. Assume that $a_{n}$ tends to $0$ and $a_{n}\sqrt{n}$ tends to
infinity.

Let $\mathbf{P}_{n}:=\frac{1}{n}\sum_{i=1}^{n}\delta _{\mathbf{X}_{i}}$
denote the empirical measure pertaining to an i.i.d. sample $\mathbf{X}_{1},%
\mathbf{X}_{2},...,\mathbf{X}_{n}.$ Write $\mathbf{M}_{n}:=\frac{1}{a_{n}}%
\left( \mathbf{P}_{n}-P\right) .$ It holds

\begin{eqnarray}
-\inf_{Q\in int(B)}\chi ^{2}(Q,P) &\leq &\lim \inf_{n}\frac{1}{na_{n}^{2}}%
\log \Pr \left( \mathbf{M}_{n}\in B\right)  \label{Sanov MDP} \\
&\leq &\lim \sup_{n}\frac{1}{na_{n}^{2}}\log \Pr \left( \mathbf{M}_{n}\in
B\right) \leq -\inf_{Q\in cl(B)}\chi ^{2}(Q,P)  \notag
\end{eqnarray}%
where the interior and closure of the set $B$ refer to the $\tau _{F}$
topology on $M_{F}.$

Consider now the asymptotic distribution of $\mathbf{X}_{1}$ conditionally
upon the sequence of events $\left( \mathbf{S}_{1}^{n}/n>a_{n}x\right) $,
so-called moderate deviation events. With $F:=B(\mathbb{R})\cup \left(
v\rightarrow v\right) $ and $B(\mathbb{R})$ the class of all bounded
measurable functions, $\left( \ref{Sanov MDP}\right) $ holds with $B$
substitued by $\Omega _{x}$ the subset of $M_{F}$ defined through%
\begin{equation*}
\Omega _{x}:=\left\{ Q:\int tdQ(t)\geq x\text{ and }\int dQ(t)=0\right\} .
\end{equation*}%
\bigskip

With $P$ the probability measure of the r.v. $\mathbf{X}_{1}$ denote $%
Q^{\ast }$ the $\chi ^{2}$ \textit{projection of }$P$ on $\Omega _{x},$
namely
\begin{equation*}
Q^{\ast }:=\arg \inf \left\{ \chi ^{2}(Q,P),Q\in \Omega _{x}\right\} .
\end{equation*}%
The set $\Omega _{x}$ is closed in $M_{F}$ $(\mathbb{R})$ equipped with the $%
\tau _{F}$ $\ $topology. Existence of a $\chi ^{2}$ projection of $P$ on a $%
\tau _{F}-$closed subset of $M(\mathbb{R})$ holds as a consequence of
Theorem 2.6\ in \cite{BroniatowskiKeziou2006} when $\int \left\vert
f\right\vert dP$ is finite for all $f$ in $F,$ which clearly holds since $%
E\left\vert \mathbf{X}_{1}\right\vert $ is finite. Uniqueness follows from
the convexity of $\Omega _{x}$ and the strict convexity of $Q\rightarrow
\chi ^{2}(Q,P).$ From (\ref{Sanov MDP}) it can easily be obtained that%
\begin{equation}
P\left( \mathbf{X}_{1}\in A/\mathcal{E}_{n,x}\right) =P(A)+a_{n}xQ^{\ast
}(A)+o\left( a_{n}\right)  \label{cond sanov mdp}
\end{equation}%
with $\mathcal{E}_{n,x}:=\left( \mathbf{S}_{1}^{n}\geq na_{n}x\right) $
which in turn yields the following

\begin{proposition}
\label{Prop Gibbs Moderate}\ With the above notation%
\begin{equation}
P\left( \mathbf{X}_{1}\in A/\mathcal{E}_{n,x}\right) =\int_{A}\pi
^{a_{n}x}(y)dy+o(1)  \label{tilted MDP}
\end{equation}
\end{proposition}

\

The proofs of (\ref{cond sanov mdp}) and of the above Proposition are
differed to the appendix. This way cannot provide an equivalent expression
for the conditional \textit{density} of $\mathbf{X}_{1}$ which requires
strong regularity assumptions. Furthermore it cannot be extended to the case
of interest here, when $\mathbf{X}_{1}$ is substituted by $\mathbf{X}%
_{1}^{k} $ for large values of $k=k_{n}$, i.e.\ when an approximation of the
law of the path $\mathbf{X}_{1}^{n}$ is needed, at least on long runs.

However the result in Proposition \ref{Prop Gibbs Moderate} is a strong
argument in favor of Ermakov's sampling scheme, namely simulating i.i.d.
r.v.'s with common density $\pi ^{a_{n}}$ in (\ref{FORM IS}).

\subsection{Density of a partial path conditioned on the exact value of the
sum}

The other way follows Zabell \cite{Zabell1980} and Diaconis and Freedman
\cite{DiaconisFreedman1988} approaches, which were developped in the range
of large deviations. See also van Camperhout and Cover \cite%
{vanCamperhoutCover1981}, who considered the density or the c.d.f. of $%
\mathbf{X}_{1}^{k}$ conditioned on the value of $\mathbf{S}_{1}^{n}$ for
fixed $k.$ It is restricted in essence to the context of the sample mean.
The sketch of the method is as follows.

The density of $\mathbf{X}_{1}$ given $\mathbf{S}_{1}^{n}=ns$ writes%
\begin{equation}
p_{\mathbf{X}_{1}/\mathbf{S}_{1}^{n}=ns}(x_{1})=\frac{p_{\mathbf{S}%
_{2}^{n}}(ns-x_{1})}{p_{\mathbf{S}_{1}^{n}}(ns)}p_{\mathbf{X}_{1}}(x_{1})
\label{cond dens}
\end{equation}%
where we used the symbol $p$ to emphasize that the $\mathbf{X}_{i}^{\prime
}s $ are i.i.d. with common density $p_{\mathbf{X}_{1}}.$ It is a known
fact, and easy to establish, that the density defined in (\ref{cond dens})
is invariant when sampling from any density of the form (\ref{tilted density}%
) instead of $p_{\mathbf{X}_{1}}.$ This yields, selecting $\alpha =s$
\begin{equation*}
p_{\mathbf{X}_{1}/\mathbf{S}_{1}^{n}=na_{n}}(x_{1})=\frac{\pi _{\mathbf{S}%
_{2}^{n}}^{s}(ns-x_{1})}{\pi _{\mathbf{S}_{1}^{n}}^{s}(ns)}\pi _{\mathbf{X}%
_{1}}^{s}(x_{1}).
\end{equation*}%
When the r.v's $\mathbf{X}_{i}$'s obey a local central limit theorem under
the sampling density $\pi _{X_{1}}^{s}$ it can be proved that
\begin{equation}
p_{\mathbf{X}_{1}/\mathbf{S}_{1}^{n}=ns}(x_{1})=\pi _{\mathbf{X}%
_{1}}^{s}(x_{1})(1+o(1))  \label{approx k=1}
\end{equation}%
as $n$ tends to $\infty .$ Diaconis and Freedman obtain such a statement
when $\mathbf{X}_{1}$ is substituted by $\mathbf{X}_{1}^{k}$ with $%
k/n\rightarrow \theta ,0\leq \theta <1.$ We will continue this approach in
the range of moderate deviations, enhancing it to the density of $\mathbf{X}%
_{1}^{k}$ with $k/n\rightarrow 1.$ Integrating with respect to the
conditional distribution of $\mathbf{S}_{1}^{n}$ under the event $\mathcal{E}%
_{n}$ provides the required approximation$.$

\bigskip

\bigskip

The scope of the present paper is to present some technique which provides
typical realisations of runs $\mathbf{X}_{1}^{k}$ under the conditional
event $\mathcal{E}_{n}$ for very large $k.$ Therefore it aims at the
exploration of the support of the distribution of $\mathbf{X}_{1}^{n}$ under
$\mathcal{E}_{n}.$ The application which is presented pertains to Importance
Sampling for the estimation of rare events probabilities through the
Adaptive Twisted IS\ scheme.

\ Section 2 of this paper is devoted to the approximation of the conditional
density of $\mathbf{X}_{1}^{k}$ under $\mathcal{E}_{n}.$ Section 3 presents
the ATIS algorithm , a number of remarks for its practical implementation,
and discusses efficiency . Section 4 is devoted to M and L estimates and
their moderate deviation probabilities. We have postponed many proofs to the
Appendix, but the main one of Section 2.

\bigskip

\section{Conditioned random walks}

\subsection{Three basic Lemmas}

Moderate deviations results for sums of i.i.d. real valued random variables
under our assumptions have been studied since the 50's by many authors. We
will make use of a \textit{local} result, due to Richter \cite{Richter1957},
which we state as

\begin{lemma}
\label{Lemma Richter local} Under the general hypotheses and notation of
this paper, when $a_{n}$ is a sequence satisfying $\lim_{n\rightarrow \infty
}a_{n}=0$ together with $\sqrt{n}a_{n}\rightarrow \infty $ it holds%
\begin{equation*}
p\left( \frac{\mathbf{S}_{1}^{n}}{n}=a_{n}\right) =\frac{\sqrt{n}\exp
-nI(a_{n})}{\sqrt{2\pi }}\left( 1+O(a_{n})\right) .
\end{equation*}
\end{lemma}

The \textit{global} counterpart of Lemma \ref{Lemma Richter local} in the
form used here is due to Jensen (see \cite{Jensen1995}, corollary 6.4.1) and
states

\begin{lemma}
\label{Lemma Jensen} Under the same hypotheses as above%
\begin{equation*}
P\left( \frac{\mathbf{S}_{1}^{n}}{n}>a_{n}\right) =\frac{\exp -nI(a_{n})}{%
\sqrt{2\pi }\sqrt{n}\psi (a_{n})}\left( 1+O(\frac{1}{\sqrt{n}})\right)
\end{equation*}%
where $\psi (a_{n}):=t_{a_{n}}s(t_{a_{n}}).$
\end{lemma}

The following known fact is used repetedly. It sets that the conditional
densities of sub-partial sums given the partial sum is invariant through any
tilting. Assume $\mathbf{X}_{1},...,\mathbf{X}_{n}$ i.i.d. with density $p$
and note $\pi ^{a}$ the corresponding tilted density for some parameter $a.$

\begin{lemma}
\label{Lemma inv conditional}For $1\leq i\leq j\leq n,$ for all $a$ in \ the
support of $P,$ for all $u$ and $s$
\begin{equation*}
p\left( \mathbf{S}_{i}^{j}=u/\mathbf{S}_{1}^{n}=s\right) =\pi ^{a}\left(
\mathbf{S}_{i}^{j}=u/\mathbf{S}_{1}^{n}=s\right) .
\end{equation*}
\end{lemma}

\bigskip

\subsection{Typical paths conditioned on their sum}

The sequence of constants $a_{n}$\ defining $A$ in (\ref{A}) and (\ref{E_n})
satisfies

\bigskip $\ (A)\ \ \ \ \ \ \ \ \ \ \ \ \ \ \ \ \ \ \ \ \ \ \ \ \ \ \ \ \ \ \
\ \ \ \ \ \left\{
\begin{array}{c}
\lim_{n\rightarrow \infty }a_{n}\sqrt{n-k}=\infty \\
\lim_{n\rightarrow \infty }\frac{na_{n}}{\sqrt{n-k}}=\infty \\
\lim_{n\rightarrow \infty }a_{n}\left( \log n\right) ^{2+\delta }=0\text{ \
for some positive }\delta \\
\lim_{n\rightarrow \infty }\frac{n-k}{n}=0%
\end{array}%
\right. $

In this section we obtain a close approximation for $p\left( \mathbf{X}%
_{1}^{k}=Y_{1}^{k}/\frac{\mathbf{S}_{1}^{n}}{n}=\sigma \right) $ for $%
k=k_{n} $ where $a_{n}$ and $k$ satisfy the following set of conditions.

The value of $\frac{\mathbf{S}_{1}^{n}}{n}$ satisfies $a_{n}\leq \sigma \leq
a_{n}+c_{n}$ with

$\ (C)\ \ \ \ \ \ \ \ \ \ \ \ \ \ \ \ \ \ \ \ \ \ \ \ \ \ \ \ \ \ \ \ \ \ \
\ \ \left\{
\begin{array}{c}
\lim_{n\rightarrow \infty }na_{n}c_{n}=\infty \\
\lim_{n\rightarrow \infty }\frac{nc_{n}}{\sqrt{n-k}}=0 \\
\lim_{n\rightarrow \infty }\frac{nc_{n}}{a_{n}\left( n-k\right) }=0 \\
\lim_{n\rightarrow \infty }\frac{\exp -na_{n}c_{n}}{a_{n}\left( \log
n\right) ^{2+\delta }}=0%
\end{array}%
\right. $

\bigskip

We denote (A1),...,(A4) , (C1),...,(C4) the above conditions.

\bigskip

It appears clearly from (\ref{FORM IS}) that the optimal choice $g=p_{%
\mathbf{X}_{1}^{n}/\mathbf{S}_{1}^{n}>na_{n}}$ need only to hold on paths $%
Y_{1}^{n}$ sampled under $g$ and not on all $\mathbb{R}^{n}.$ \ In a similar
way the approximation of the optimal density need to be realized only when
evaluated on samples $Y_{1}^{n}(l)$ generated according to this
approximation, and approximation of $p_{\mathbf{X}_{1}^{n}/\mathcal{E}}$ on
the entire space $\mathbb{R}^{n}$ is not needed$.$ The approximation of $%
\mathfrak{p}_{n}$ by such a density $g_{n}$ is difficult to obtain on
realizations under $g_{n}$ and much easier under $p_{\mathbf{X}_{1}^{n}/%
\mathbf{S}_{1}^{n}>na_{n}}.$ The following Lemma proves that approximating $%
\mathfrak{p}_{n}$ by $g_{n}$ under $\mathfrak{p}_{n}$ is similar to
approximating $\mathfrak{p}_{n}$ by $g_{n}$ under $g_{n}.$

Let $\mathfrak{R}_{n}$ and $\mathfrak{S}_{n}$ denote two p.m's on $\mathbb{R}%
^{n}$ with respective densities $\mathfrak{r}_{n}$ and $\mathfrak{s}_{n}.$

\begin{lemma}
\label{Lemmafrompntogn} Suppose that for some sequence $%
\varepsilon _{n}$ $\ $which tends to $0$ as $n$ tends to infinity%
\begin{equation}
\mathfrak{r}_{n}\left( Y_{1}^{n}\right) =\mathfrak{s}_{n}\left(
Y_{1}^{n}\right) \left( 1+o_{\mathfrak{R}_{n}}(\varepsilon _{n})\right)
\label{pnequivgnunderpn}
\end{equation}%
as $n$ tends to $\infty .$ Then
\begin{equation}
\mathfrak{s}_{n}\left( Y_{1}^{n}\right) =\mathfrak{r}_{n}\left(
Y_{1}^{n}\right) \left( 1+o_{\mathfrak{S}_{n}}(\varepsilon _{n})\right) .
\label{gnequivpnundergn}
\end{equation}
\end{lemma}
\begin{proof}
Denote
\begin{equation*}
A_{n,\varepsilon _{n}}:=\left\{ y_{1}^{n}:(1-\varepsilon _{n})\mathfrak{s}%
_{n}\left( y_{1}^{n}\right) \leq \mathfrak{r}_{n}\left( y_{1}^{n}\right)
\leq \mathfrak{s}_{n}\left( y_{1}^{n}\right) (1+\varepsilon _{n})\right\} .
\end{equation*}%
It holds for all positive $\delta $%
\begin{equation*}
\lim_{n\rightarrow \infty }I(n,\delta )=1
\end{equation*}%
where
\begin{equation*}
I(n,\delta ):=\int \mathbf{1}_{A_{n,\delta \varepsilon _{n}}}\left(
y_{1}^{n}\right) \frac{\mathfrak{r}_{n}\left( y_{1}^{n}\right) }{\mathfrak{s}%
_{n}(y_{1}^{n})}\mathfrak{s}_{n}(y_{1}^{n})dy_{1}^{n}.
\end{equation*}%
Since
\begin{equation*}
I(n,\delta )\leq (1+\delta \varepsilon _{n})\mathfrak{S}_{n}\left(
A_{n,\delta \varepsilon _{n}}\right)
\end{equation*}%
it follows that
\begin{equation*}
\lim_{n\rightarrow \infty }\mathfrak{S}_{n}\left( A_{n,\delta \varepsilon
_{n}}\right) =1,
\end{equation*}%
which proves the claim.
\end{proof}

This shows that the approximation of $\mathfrak{p}_{n}$ need not to be
achieved on the whole space $\mathbb{R}^{n}$ but only on \textit{typical
paths} under the conditionning event $\mathcal{E}_{n}.$ It appears that\
such a sharp approximation is possible on quite long portions $Y_{1}^{k}$\
of sample paths generated under $\mathfrak{P}_{n}$,$\ $when $k$ tends to $%
\infty $ \ together with $n$ and $k/n$ goes to $1.$ \bigskip

Let $\sigma $ such that $a_{n}\leq \sigma \leq b_{n}$ with $b_{n}-a_{n}$
small enough.\ We prove that the sequence of conditional densities $p\left(
\mathbf{X}_{1}^{k}=Y_{1}^{k}/\mathbf{S}_{1}^{n}=n\sigma \right) $ is closely
approximated by a sequence of suitably modified tilted densities when
evaluated at $Y_{1}^{k},$ a realization under the density $\mathfrak{p}_{n}.$
This is the scope of Proposition \ref{Prop approx local cond density}
hereunder.\ The size of $b_{n}-a_{n}$ is such that $\mathfrak{p}_{n}\left(
Y_{1}^{k}\right) $ can be substituted by an integral of $p\left( \mathbf{X}%
_{1}^{k}=Y_{1}^{k}/\mathbf{S}_{1}^{n}=n\sigma \right) $ with respect to the
distribution of $\mathbf{S}_{1}^{n}$ conditionally on $\left( \mathbf{S}%
_{1}^{n}\in \left( na_{n},nb_{n}\right) \right) .$ This is the scope of
Proposition \ref{Prop p_n equiv g under S_n>na_n}.

Define $\Sigma _{1}^{i}:=Y_{1}$ $+...+Y_{i}$ and $t_{i,n}$ through%
\begin{equation}
m(t_{i,n})=m_{i,n}:=\frac{n}{n-i}\left( \sigma -\frac{\Sigma _{1}^{i}}{n}%
\right)  \label{def m_i,n}
\end{equation}%
\begin{equation*}
s_{i,n}^{2}:=\frac{d^{2}}{dt^{2}}\left( \log E_{\pi ^{m_{i,n}}}\exp t\mathbf{%
X}_{1}\right) \left( 0\right)
\end{equation*}%
and%
\begin{equation*}
\mu _{3}^{(i,n)}:=\frac{d^{3}}{dt^{3}}\left( \log E_{\pi ^{m_{i,n}}}\exp t%
\mathbf{X}_{1}\right) \left( 0\right)
\end{equation*}%
which are the variance and the kurtosis of $\pi ^{m_{i,n}},$ reflecting the
corresponding characteristics of $p$, since $t_{i,n}$ is close to $0$ as
shown in the following result. \

\begin{lemma}
\label{Lemmaorderoftinetc} Let $\sigma $ belong to $(a_{n},b_{n})$ and
assume that (A) holds together with (C2) and (C3). Then under $\mathfrak{P}%
_{n}$ , $t_{i,n}$ tends to $0$, $s_{i,n}^{2}$ tends to $1$ and $\mu
_{3}^{(i,n)}$ tends to the third centered moment of $p$ uniformly upon $%
\sigma $ in $(a_{n},b_{n}).$
\end{lemma}

\begin{proof}
Write
\begin{equation*}
m(t_{i,n})=\frac{n}{n-i}(\sigma -a_{n})+\frac{n}{n-i}\left( a_{n}-\frac{%
\Sigma _{1}^{i}}{n}\right)
\end{equation*}%
which goes to $0$ under $\mathfrak{P}_{n}$ uniformly upon $\sigma $ under
(C2) and (C3) where we used Lemma \ref{Lemma m_i,n under conditioning};
therefore $t_{i,n}$ goes to $0$ uniformly in $\sigma $ which concludes the
proof.
\end{proof}

\bigskip

The following density $g_{\sigma }(y_{1}^{k})$ defined in (\ref{g_s}) on $%
\mathbb{R}^{k}$ provides the sharp approximation of $p(\mathbf{X}%
_{1}^{k}=y_{1}^{k}/\mathbf{S}_{1}^{n}=n\sigma ).$ This density is defined on
$\mathbb{R}^{k}$ as a product of conditional densities which are set in the
following displays. It only approximates $p(\mathbf{X}_{1}^{k}=y_{1}^{k}/%
\mathbf{S}_{1}^{n}=n\sigma )$ on typical vectors $y_{1}^{k}$ which are
realizations of $\mathbf{X}_{1}^{k}$ under $\mathcal{E}_{n}.$ Chose any
density $g_{0}(y_{1})$ (for convenience denoted $g_{0}(y_{1}/y_{0})$ in (\ref%
{g_s}).and for $1\leq i\leq k-1$ define recursively the sequence of
conditional densities $g_{i}(y_{i+1}/y_{1}^{i})$ through
\begin{equation*}
g_{0}(y_{1})=\pi ^{\sigma }(y_{1})
\end{equation*}%
and
\begin{equation}
g_{i}(y_{i+1}/y_{1}^{i})=\frac{\exp \left( y_{i+1}\left( t_{i,n}+\frac{\mu
_{3}^{(i,n)}}{2s_{i,n}^{4}\left( n-i-1\right) }\right) -y_{i+1}^{2}/\left(
2s_{i,n}^{2}\left( n-i-1\right) \right) \right) p(y_{i+1})}{K_{i}(y_{1}^{i})}
\label{g_i}
\end{equation}%
a density on $\mathbb{R}$ , with $t_{i,n}$ the unique solution of the
equation%
\begin{equation*}
m(t_{i,n})=\frac{n}{n-i}\left( \sigma -\frac{s_{1}^{i}}{n}\right)
\end{equation*}%
where $s_{1}^{i}:=y_{1}+...+y_{i}.$ The normalizing factor $K_{i}(y_{1}^{i})$
is
\begin{equation}
K_{i}(y_{1}^{i})=\int \exp \left( x\left( t_{i,n}+\frac{\mu _{3}^{(i,n)}}{%
2s_{i,n}^{4}\left( n-i-1\right) }\right) -x^{2}/\left( 2s_{i,n}^{2}\left(
n-i-1\right) \right) \right) p(x)dx.  \label{K_i}
\end{equation}%
Define $g_{\sigma }$ the density on $\mathbb{R}^{k}$ through%
\begin{equation}
g_{\sigma }(y_{1}^{k}):=\prod_{i=0}^{k-1}g_{i}(y_{i+1}/y_{1}^{i}).
\label{g_s}
\end{equation}%
The definition in,(\ref{g_i}) can also be stated as
\begin{equation*}
g_{i+1}(y_{i+1}/x_{1}^{i})=C_{i}p(y_{i+1})\mathfrak{n}\left(
ab,a,y_{i+1}\right)
\end{equation*}%
where $\mathfrak{n}\left( \mu ,\sigma ^{2},x\right) $ is the normal density
with mean $\mu $ and variance $\sigma ^{2}$ at $x$. Here
\begin{equation*}
a=s_{i,n}^{2}\left( n-i-1\right)
\end{equation*}%
\begin{equation*}
b=t_{i,n}+\frac{\mu _{3}^{(i,n)}}{2s_{i,n}^{4}\left( n-i-1\right) }
\end{equation*}%
and the constant $C_{i}$ is $\left( K_{i}(y_{1}^{i})\right) ^{-1}$. This
form is appropriate for the simulation.

The density $g_{i}(y_{i+1}/y_{1}^{i})$ is a slight modification from $\pi
^{m(t_{i,n})}$. It approximates sharply p$\left( \mathbf{X}_{i+1}=y_{i+1}/%
\mathbf{S}_{1}^{n}=n\sigma ,y_{1}^{i}\right) $ . \ For small values of $i$,
the contribution of $y_{i+1}\frac{\mu _{3}^{(i,n)}}{2s_{i,n}^{4}\left(
n-i-1\right) }$ and of $y_{i+1}^{2}/\left( 2s_{i,n}^{2}\left( n-i-1\right)
\right) $ is small and $g_{i}(y_{i+1}/y_{1}^{i})$ fits nearly with $\pi
^{a_{n}}(y_{i+1})$, when $\sigma $ is close to $a_{n}$, which is in
accordance both with Diaconis and Freedman's approximation when translated
in the moderate deviation range and with Ermakov's IS scheme.

\begin{remark}
When the $X_{i}$'s are i.i.d. normal then $%
g_{i}(y_{i+1}/y_{1}^{i})=p(y_{i+1}/y_{1}^{i},\frac{S_{n}}{n}=\sigma )$ for
all $i.$
\end{remark}

We then have

\begin{proposition}
\label{Prop approx local cond density}Set $\sigma $ with $a_{n}\leq \sigma
\leq b_{n}$ and assume (A) together with (C2) and (C3). Let $Y_{1}^{n}$ be a
sample with distribution $\mathfrak{P}_{n}.$ Then uniformly upon $\sigma $
\begin{equation}
p(\mathbf{X}_{1}^{k}=Y_{1}^{k}/\mathbf{S}_{1}^{n}=n\sigma )=g_{\sigma
}(Y_{1}^{k})(1+o_{\mathfrak{P}_{n}}(a_{n}\left( \log n\right) ^{2+\delta })).
\label{local approx under exact value}
\end{equation}
\end{proposition}

\begin{proof}
The proof uses a Bayes formula to write $p(\mathbf{X}_{1}^{k}=Y_{1}^{k}/%
\mathbf{S}_{1}^{n}=n\sigma )$ as a product of $k$ conditional densities of
individual terms of the trajectory evaluated at $Y_{1}^{k}$, and the
invariance property stated in Lemma \ref{Lemma inv conditional}.\ Each term
of this product is approximated through an Edgeworth expansion which
together with the three preceeding lemmas, conclude the proof. \ It holds
\begin{eqnarray}
p(\mathbf{X}_{1}^{k} &=&Y_{1}^{k}/\mathbf{S}_{1}^{n}=n\sigma )=p(\mathbf{X}%
_{1}=Y_{1}/\mathbf{S}_{1}^{n}=n\sigma )  \label{joint density} \\
\prod_{i=1}^{k-1}p(\mathbf{X}_{i+1} &=&Y_{i+1}/\mathbf{X}%
_{1}^{i}=Y_{1}^{i},\mathbf{S}_{1}^{n}=n\sigma ) \\
&=&\prod_{i=0}^{k-1}p\left( \mathbf{X}_{i+1}=Y_{i+1}/\mathbf{S}%
_{i+1}^{n}=n\sigma -\Sigma _{1}^{i}\right)  \notag
\end{eqnarray}%
by the independence of the r.v's $\mathbf{X}_{i}$; we have set $%
S_{1}^{0}:=0. $ By Lemma \ref{Lemma inv conditional}%
\begin{eqnarray*}
&&p\left( \mathbf{X}_{i+1}=Y_{i+1}/\mathbf{S}_{i+1}^{n}=n\sigma -\Sigma
_{1}^{i}\right) \\
&=&\pi ^{m_{i,n}}\left( \mathbf{X}_{i+1}=Y_{i+1}/\mathbf{S}%
_{i+1}^{n}=n\sigma -\Sigma _{1}^{i}\right) \\
&=&\pi ^{m_{i,n}}\left( \mathbf{X}_{i+1}=Y_{i+1}\right) \frac{\pi
^{m_{i,n}}\left( \mathbf{S}_{i+2}^{n}=n\sigma -\Sigma _{1}^{i+1}\right) }{%
\pi ^{m_{i,n}}\left( \mathbf{S}_{i+1}^{n}=n\sigma -\Sigma _{1}^{i}\right) }
\end{eqnarray*}%
where we used Bayes formula and the independence of the $\mathbf{X}_{j}$'s
under $\pi ^{m_{i,n}}.$ A precise evaluation of the dominating terms in this
lattest expression is needed in order to handle the product (\ref{joint
density}).

Under the sequence of densities $\pi ^{m_{i,n}}$ the i.i.d. r.v's $\mathbf{X}%
_{i+1},...,\mathbf{X}_{n}$ define a triangular array which satisfies a local
central limit theorem, and an Edgeworth expansion. Under $\pi ^{m_{i,n}}$, $%
\mathbf{X}_{i+1}$ has expectation $m_{i,n}$ and variance $s_{i,n}^{2}.$
Center and normalize both the numerator and denominator in the fraction
which appears in the last display. Denote $\overline{\pi _{n-i-1}^{m_{i,n}}}$
the density of the normalized partial sum $\left( \mathbf{S}%
_{i+2}^{n}-(n-i-1)m_{i,n}\right) /\left( s_{i,n}\sqrt{n-i-1}\right) $ when
the summands are i.i.d. with common density $\pi ^{m_{i,n}}.$ Hence,
evaluating both $\overline{\pi _{n-i-1}^{m_{i,n}}}$ and its normal
approximation at point $Y_{i+1},$
\begin{eqnarray}
&&p\left( \mathbf{X}_{i+1}=Y_{i+1}/\mathbf{S}_{i+1}^{n}=n\sigma -\Sigma
_{1}^{i}\right)  \label{condTilt} \\
&=&\frac{\sqrt{n-i}}{\sqrt{n-i-1}}\pi ^{m_{i,n}}\left( \mathbf{X}%
_{i+1}=Y_{i+1}\right) \frac{\overline{\pi _{n-i-1}^{m_{i,n}}}\left( \left(
m_{i,n}-Y_{i+1}\right) /s_{i,n}\sqrt{n-i-1}\right) }{\overline{\pi
_{n-i}^{m_{i,n}}}(0)}.  \notag
\end{eqnarray}%
The sequence of densities $\overline{\pi _{n-i-1}^{m_{i,n}}}$ converges
pointwise to the standard normal density under the assumptions, when $n-i$
tends to infinity, i.e. when $n-k_{n}$ tends to infinity, and an Edgeworth
expansion to the order 5 is performed for the numerator and the denominator.

Set $Z_{i+1}:=\left( m_{i,n}-Y_{i+1}\right) /s_{i,n}\sqrt{n-i-1}.$ Using
Lemma \ref{Lemmaminunderconditioning} we have
\begin{equation}
m_{i,n}-Y_{i+1}=a_{n}-Y_{i+1}+\frac{n\left( \sigma -a_{n}\right) }{n-i-1}+O_{%
\mathfrak{P}_{n}}\left( \frac{1}{\sqrt{n-i}}\right) .  \label{m_i,n-Y_i+1}
\end{equation}%
It then holds%
\begin{eqnarray}
\overline{\pi _{n-i-1}^{m_{i,n}}}\left( Z_{i+1}\right) &=&\mathfrak{n}%
(Z_{i+1})\left[
\begin{array}{c}
1+\frac{1}{\sqrt{n-i-1}}P_{3}(Z_{i+1})+\frac{1}{n-i-1}P_{4}(Z_{i+1}) \\
+\frac{1}{\left( n-i-1\right) ^{3/2}}P_{5}(Z_{i+1})%
\end{array}%
\right]  \label{Hermite} \\
&&+O_{\mathfrak{P}_{n}}\left( \frac{1}{\left( n-i-1\right) ^{3/2}}\right) .
\notag
\end{eqnarray}
We perform an expansion in $\mathfrak{n}(Z_{i+1})$ up to the order $3,$ with
a first order term $\mathfrak{n}\left( -Y_{i+1}/\left( s_{i,n}\sqrt{n-i-1}%
\right) \right) ,$ namely%
\begin{eqnarray}
\mathfrak{n}(Z_{i+1}) &=&\mathfrak{n}\left( -Y_{i+1}/\left( s_{i,n}\sqrt{%
n-i-1}\right) \right)  \label{approx gauss} \\
&&\left(
\begin{array}{c}
1-\frac{Y_{i+1}m_{i,n}}{s_{i,n}^{2}\left( n-i-1\right) }+\frac{m_{i,n}^{2}}{%
2s_{i,n}^{2}\left( n-i-1\right) }\left( \frac{Y_{i+1}^{2}}{s_{i,n}^{2}\left(
n-i-1\right) }-1\right) \\
+\frac{m_{i,n}^{3}}{6s_{i,n}^{3}\left( n-i-1\right) ^{3/2}}\frac{\mathfrak{n}%
^{(3)}\left( \frac{Y^{\ast }}{\left( s_{i,n}\sqrt{n-i-1}\right) }\right) }{%
\mathfrak{n}\left( -Y_{i+1}/\left( s_{i,n}\sqrt{n-i-1}\right) \right) }%
\end{array}%
\right)
\end{eqnarray}%
where $Y^{\ast }=\frac{1}{s_{i,n}\sqrt{n-i-1}}(-Y_{i+1}+\theta m_{i,n})$
with $\left\vert \theta \right\vert <1.$ Only the first order term is
relevant when handling the conditional density of the sub trajectory $%
Y_{1}^{k}.$

Write
\begin{equation*}
m_{i,n}=\frac{n}{n-i-1}\left( a_{n}-\frac{\Sigma _{1}^{i}}{n}\right) +\frac{n%
}{n-i-1}\left( \sigma -a_{n}\right)
\end{equation*}%
and use Lemmas \ref{Lemma m_i,n under conditioning} and \ref{Lemma max Y_i
under E_n} to obtain%
\begin{eqnarray}
\frac{\left\vert Y_{i+1}m_{i,n}\right\vert }{s_{i,n}^{2}\left( n-i-1\right) }
&=&\frac{O_{\mathfrak{P}_{n}}(\log k)}{n-i-1}\left( a_{n}+O_{\mathfrak{P}%
_{n}}\left( \frac{1}{\sqrt{n-i}}\right) \right) (1+o_{\mathfrak{P}_{n}}(1))
\label{control 1} \\
&&+n\left( \sigma -a_{n}\right) \frac{O_{\mathfrak{P}_{n}}(\log k)}{\left(
n-i-1\right) ^{2}}(1+o_{\mathfrak{P}_{n}}(1))  \notag
\end{eqnarray}%
and%
\begin{eqnarray}
\frac{m_{i,n}^{2}}{s_{i,n}^{2}\left( n-i-1\right) } &=&\frac{1}{n-i-1}\left(
a_{n}+O_{\mathfrak{P}_{n}}\left( \frac{1}{\sqrt{n-i}}\right) \right)
^{2}(1+o_{\mathfrak{P}_{n}}(1))  \label{control 2} \\
&&+\frac{n^{2}}{\left( n-i-1\right) ^{3}}\left( \sigma -a_{n}\right)
^{2}(1+o_{\mathfrak{P}_{n}}(1))  \notag \\
&&+2\frac{n\left( \sigma -a_{n}\right) }{\left( n-i-1\right) ^{2}}\left(
a_{n}+O_{\mathfrak{P}_{n}}\left( \frac{1}{\sqrt{n-i}}\right) \right) (1+o_{%
\mathfrak{P}_{n}}(1)).  \notag
\end{eqnarray}%
where the $1+o_{\mathfrak{P}_{n}}(1)$ terms stem from the convergence of $%
s_{i,n}^{2}$ to $1$ by Lemma \ref{Lemmaorderoftinetc}$.$ Assuming (C2)
it follows that
\begin{equation*}
\frac{\left\vert Y_{i+1}m_{i,n}\right\vert }{s_{i,n}^{2}\left( n-i-1\right) }%
=\frac{O_{\mathfrak{P}_{n}}(\log k)}{n-i-1}\left( a_{n}+O_{\mathfrak{P}%
_{n}}\left( \frac{1}{\sqrt{n-i}}\right) \right) (1+o_{\mathfrak{P}_{n}}(1))
\end{equation*}%
and%
\begin{equation*}
\frac{m_{i,n}^{2}}{s_{i,n}^{2}\left( n-i-1\right) }=\frac{1}{n-i-1}\left(
a_{n}+O_{\mathfrak{P}_{n}}\left( \frac{1}{\sqrt{n-i}}\right) \right)
^{2}(1+o_{\mathfrak{P}_{n}}(1))
\end{equation*}%
which yields
\begin{equation}
\mathfrak{n}(Z_{i+1})=\mathfrak{n}\left( -Y_{i+1}/\left( s_{i,n}\sqrt{n-i-1}%
\right) \right) \left( 1+O_{\mathfrak{P}_{n}}\left( \frac{a_{n}\log k}{n-i}%
\right) \right) .  \label{n(Z_i+1)}
\end{equation}

The Hermite polynomials depend upon the moments of the underlying density $%
\pi ^{m_{i,n}}.$ Since $\overline{\pi _{1}^{m_{i,n}}}$ has expectation $0$
and variance $1$ the terms corresponding to $P_{1}$ \ and $P_{2}$ vanish. Up
to the order $4$ the polynomials write $P_{3}(x)=\frac{\mu _{3}^{(i)}}{%
6\left( s_{i,n}\right) ^{3}}(x^{3}-3x)$, $P_{4}(x)=\frac{\mu _{3}^{(i,n)}}{%
72\left( s_{i,n}\right) ^{6}}(x^{3}-3x)+\frac{\mu _{4}^{(i,n)}-3\left(
s_{i,n}\right) ^{4}}{24\left( s_{i,n}\right) ^{4}}\left(
x^{4}+6x^{2}-3\right) $.

In order to obtain a development of the polynomial bracket in (\ref{Hermite}%
) in terms of powers of $(n-i)$ only the term in $x$ from $P_{3}$ and the
constant term from $P_{4}$ are relevant. It holds%
\begin{eqnarray*}
\frac{P_{3}(Z_{i+1})}{\sqrt{n-i-1}} &=&-\frac{\mu _{3}^{(i,n)}}{%
2s_{i,n}^{4}\left( n-i-1\right) }\left( a_{n}-Y_{i+1}\right) -\frac{\mu
_{3}^{(i,n)}}{2s_{i,n}^{4}}\frac{n\left( \sigma -a_{n}\right) }{\left(
n-i-1\right) ^{2}} \\
&&\text{ }-\frac{\mu _{3}^{(i,n)}\left( m_{i,n}-Y_{i+1}\right) ^{3}}{6\left(
s_{i,n}\right) ^{6}(n-i-1)^{2}}+O_{\mathfrak{P}_{n}}\left( \frac{1}{\left(
n-i\right) ^{3/2}}\right) .
\end{eqnarray*}%
When (C3) holds then
\begin{eqnarray}
\frac{P_{3}(Z_{i+1})}{\sqrt{n-i-1}} &=&-\frac{\mu _{3}^{(i,n)}}{%
2s_{i,n}^{4}\left( n-i-1\right) }\left( a_{n}-Y_{i+1}\right) +O_{\mathfrak{P}%
_{n}}\left( \frac{1}{\left( n-i\right) ^{3/2}}\right)  \label{P3} \\
&=&\frac{\mu _{3}^{(i,n)}}{2s_{i,n}^{4}\left( n-i-1\right) }Y_{i+1}+O_{%
\mathfrak{P}_{n}}\left( \frac{1}{\left( n-i\right) ^{3/2}}\right)
+a_{n}O\left( \frac{1}{n-i}\right) .  \notag
\end{eqnarray}%
For the term of order $4$ it holds
\begin{equation}
\frac{P_{4}(Z_{i+1})}{n-i-1}=\frac{1}{n-i-1}\left( \frac{1}{12s_{i,n}^{3}}%
P_{3}(Z_{i+1})+\frac{\mu _{4}^{(i,n)}-3s_{i,n}^{4}}{24s_{i,n}^{4}\left(
n-i-1\right) }\left( Z_{i+1}^{4}+6Z_{i+1}^{2}-3\right) \right) .  \label{P4}
\end{equation}%
When (C2) and (C3) hold it follows that
\begin{equation*}
\frac{P_{4}(Z_{i+1})}{n-i-1}=-\frac{\mu _{4}^{(i,n)}-3s_{i,n}^{4}}{%
8s_{i,n}^{4}\left( n-i-1\right) }+O_{\mathfrak{P}_{n}}\left( \frac{1}{\left(
n-i-1\right) ^{3/2}}\right) .
\end{equation*}%
The fifth term in the expansion plays no role in the asymptotics, under (A).
To sum up and using (A) and Lemma \ref{Lemma max Y_i under E_n} we get%
\begin{equation}
\overline{\pi _{n-i-1}^{m_{i,n}}}\left( Z_{i+1}\right) =\mathfrak{n}\left(
-Y_{i+1}/\left( s_{i,n}\sqrt{n-i-1}\right) \right) \left(
\begin{array}{c}
1+\frac{\mu _{3}^{(i,n)}}{2s_{i,n}^{4}\left( n-i-1\right) }Y_{i+1} \\
-\frac{\mu _{3}^{(i,n)}-s_{i,n}^{4}}{8s_{i,n}^{4}(n-i-1)}+O_{\mathfrak{P}%
_{n}}\left( \frac{a_{n}\log n}{n-i}\right)%
\end{array}%
\right) .  \label{num approx fixed i}
\end{equation}

Turn back to (\ref{condTilt}) and do the same Edgeworth expansion in the
demominator, which writes%
\begin{equation}
\overline{\pi _{n-i}^{m_{i,n}}}(0)=\mathfrak{n}(0)\left( 1-\frac{\mu
_{3}^{(i,n)}-s_{i,n}^{4}}{8s_{i,n}^{4}(n-i)}\right) +O_{\mathfrak{P}%
_{n}}\left( \frac{1}{\left( n-i\right) ^{3/2}}\right) .  \label{PI 0}
\end{equation}%
Summarizing and using both (\ref{P3}) and (\ref{P4}) we obtain
\begin{eqnarray}
&&p\left( \mathbf{X}_{i+1}=Y_{i+1}/\mathbf{S}_{i+1}^{n}=n\sigma -\Sigma
_{i}^{n}\right)  \label{ratio fixed i} \\
&=&\frac{\sqrt{n-i}}{\sqrt{n-i-1}}\exp \left( Y_{i+1}\left( t_{i,n}+\frac{%
\mu _{3}^{(i,n)}}{2s_{i,n}^{2}\left( n-i-1\right) }\right)
-Y_{i+1}^{2}/\left( 2s_{i,n}^{2}\left( n-i-1\right) \right) \right)  \notag
\\
&&\frac{p(Y_{i+1})}{\phi (t_{i,n})}\left( 1+O_{\mathfrak{P}_{n}}\left( \frac{%
a_{n}\log n}{n-i}\right) \right) .  \notag
\end{eqnarray}

The term $\exp -Y_{i+1}^{2}/2s_{i,n}^{2}(n-i-1)$ in $g_{i}(Y_{i+1}/Y_{i})$
comes from the ratio of the two gaussian densities $\mathfrak{n}(Z_{i+1})$
and $\mathfrak{n}(0).$ Taking logarithms and using standard calculus
provides the result in (\ref{g_i}); indeed the constant term $-\frac{\mu
_{3}^{(i,n)}-3s_{i,n}^{4}}{8s_{i,n}^{4}(n-i-1)}$ in (\ref{P4}) combines with
the corresponding one in (\ref{PI 0}) to produce a term of order $O_{%
\mathfrak{P}_{n}}\left( \frac{1}{\left( n-i\right) ^{2}}\right) $ whose sum
is $O_{\mathfrak{P}_{n}}\left( \frac{1}{\left( n-k\right) }\right) =o_{%
\mathfrak{P}_{n}}\left( a_{n}\left( \log n\right) ^{2+\delta }\right) $.

We now prove that $K_{i}$ as defined in (\ref{K_i}) satisfies
\begin{equation}
K_{i}(Y_{1}^{i})=\phi (t_{i,n})\left( 1-\frac{1}{2(n-i-1)}\right) +O_{%
\mathfrak{P}_{n}}\left( \frac{1}{(n-i)^{3/2}}\right) .
\label{Kiapproximation}
\end{equation}
This will conclude the proof.

Use the classical bounds
\begin{equation*}
1-u+\frac{u^{2}}{2}-\frac{u^{3}}{6}\leq e^{-u}\leq 1-u+\frac{u^{2}}{2}
\end{equation*}%
to obtain on both sides of the above inequalities the second order
approximation of $K_{i}(Y_{1}^{i}).$ The upper bound is
\begin{eqnarray*}
K_{i}(Y_{1}^{i}) &\leq &\phi \left( t_{i,n}\right) +\frac{\mu _{3}^{(i,n)}}{%
2s_{i,n}^{2}\left( n-i-1\right) }\phi ^{\prime }\left( t_{i,n}\right) +\frac{%
\mu _{3}^{(i,n)2}}{(2)^{2}s_{i,n}^{4}\left( n-i-1\right) ^{2}}\phi "\left(
t_{i,n}\right) \\
&&-\frac{1}{2s_{i,n}^{2}(n-i-1)}\left[ \phi "\left( t_{i,n}\right) +\frac{%
\mu _{3}^{(i,n)}}{2s_{i,n}^{2}\left( n-i-1\right) }\phi ^{(3)}(t_{i,n})%
\right] .\end{eqnarray*}%
The lower bound is the same up to order 2 and the third order term plays no
role.

Use Lemma \ref{Lemmaorderoftinetc} to conclude, making a Taylor
expansion in $\phi \left( t_{i,n}\right) ,$ $\phi ^{\prime }\left(
t_{i,n}\right) $ and $\phi "\left( t_{i,n}\right) .$ The dominating terms
are due to $\phi \left( t_{i,n}\right) $ and $\frac{1}{2s_{i,n}^{2}(n-i-1)}%
\phi "\left( t_{i,n}\right) $ which yield the $1-\frac{1}{2(n-i-1)}$ term in
(\ref{Kiapproximation}). The other terms are indeed $$O_{\mathfrak{P}_{n}}\left( \frac{1}{(n-i)^{3/2}}\right) $$
using Lemma \ref{Lemmaminunderconditioning},
leading to (\ref{Kiapproximation}). Hence (\ref{ratiofixed i}) writes as
\begin{equation*}
(\ref{ratio fixed i})=g_{i}(Y_{i+1}/Y_{1}^{i})\left( 1+O_{\mathfrak{P}%
_{n}}\left( \frac{a_{n}\log n}{n-i}\right) \right) .
\end{equation*}
Putting the pieces together yields under (A)
\begin{equation*}
p(\mathbf{X}_{1}^{k}=Y_{1}^{k}/\mathbf{S}_{1}^{n}=n\sigma )=\left( 1+o_{%
\mathfrak{P}_{n}}\left( a_{n}\left( \log n\right) ^{2+\delta }\right)
\right) \prod_{i=1}^{k}g_{i}(Y_{i+1}/Y_{1}^{i}).
\end{equation*}%
Uniformity upon $\sigma $ is a consequence of Lemma \ref{Lemmaorderoftinetc}. This closes the proof of the Proposition.
\end{proof}
\begin{remark}
When the $X_{i}$'s are i.i.d. normal, then the result in the above
Proposition holds with $k=n$ stating that $p(\mathbf{X}_{1}^{n}=x_{1}^{n}/%
\mathbf{S}_{1}^{n}=n\sigma )=g_{\sigma }\left( x_{1}^{n}\right) $ for all $%
x_{1}^{n}$ in $\mathbb{R}^{n}$ .
\end{remark}

\begin{remark}
The density in (\ref{g_i}) is a slight modification of $\pi ^{m_{i,n}}.$
However second order terms are required here in order to handle the
approximation of the density of $\mathbf{X}_{i+1\text{ }}$conditioned upon $%
\mathbf{X}_{1}^{i}$ and $\mathbf{S}_{1}^{n}/n.$ The modification from $\pi
^{m_{i,n}}$ to $g_{i}$ is a small shift in the location parameter, which
reflects the asymmetry of the underlying distribution $p,$ and a change in
the variance : large values of $\ \mathbf{X}_{i+1}$ have smaller weight for
large $i,$ which is to say that the distribution of $\ \mathbf{X}_{i+1}$
tends to concentrate around $m_{i,n}$ as $i$ approaches $k.$
\end{remark}

\begin{remark}
The "moderate deviation" case is typically $a_{n}=n^{-\tau },$ for $\tau $
in $\left( 0,1/2\right) .$ In this case the condition $a_{n}\left( \log
n\right) ^{2+\delta }\rightarrow 0$ holds for all values of $\tau .$ The
other case is when $a_{n}$ is "nearly constant", in the range $a_{n}=\left(
\log n\right) ^{-\gamma }$,$\gamma <2,$ decreasing very slowly to $0,$ with $%
\gamma >2+\delta ,\delta >0.$
\end{remark}

\begin{remark}
\label{Remark Edgeworth array}In Lemmas \ref{Lemma m_i,n under conditioning}
and \ref{Lemma max X_i under conditioning} , as in the previous Proposition,
we use an Edgeworth expansion for the density of the normalized sum of the $%
n-$th row of some triangular array of row-wise independent r.v's with common
density. Consider the i.i.d. r.v's $\mathbf{X}_{1},...,\mathbf{X}_{n}$ with
common density $\pi ^{\sigma }(x)$ where $\sigma $ may depend on $n$ but
remains bounded$.$ The Edgeworth expansion pertaining to $\overline{\pi
_{n}^{\sigma }}$ can be derived following closely the proof given for
example in \cite{Feller1971}, pp 532 and followings substituting the
cumulants of $p$ by those of $\pi ^{\sigma }$. Denote $\varphi _{\sigma }(z)$
the characteristic function of $\pi ^{\sigma }(x).$ Clearly for any $\delta
>0$ there exists $q_{\sigma ,\delta }<1$ such that $\left\vert \varphi
_{\sigma }(z)\right\vert <$ $q_{\sigma ,\delta }$ and since $a_{n}$ is
bounded, $\sup_{n}q_{\sigma ,\delta }<1.$ Therefore the inequality (2.5) in
\cite{Feller1971} p533 holds. With $\psi _{n}$ defined as in \cite%
{Feller1971} (2.6) holds with $\varphi $ replaced by $\varphi _{\sigma }$
and $\sigma $ by $s(t_{\sigma });$ (2.9) holds, which completes the proof of
the Edgeworth expansion in the simple case. The proof goes in the same way
for higher order expansions. This justifies our argument in the Lemmas cited
above. In the proofs of Proposition \ref{Prop approx local cond density} we
made use of such expansions when the r.v's $\mathbf{X}_{i+1},...,\mathbf{X}%
_{n}$ are i.i.d. with common density $\pi ^{m_{i,n}}(x).$The same argument
as sketched hereabove applies in this case also.
\end{remark}

\bigskip

\subsection{Conditioning on final events $\mathcal{E}_{n}$}

Let $\mathbf{\ T:=}$ $\mathbf{S}_{1}^{n}/n$ with distribution under the
conditioning event $\mathcal{E}_{n}.$ Hence for any Borel set $A$
\begin{equation}
P\left( \mathbf{T}\in A\right) =\mathfrak{P}_{n}\left( \frac{\mathbf{S}%
_{1}^{n}}{n}\in A\right) .  \label{T}
\end{equation}%
The distribution of $\mathbf{T}$ is concentrated on a small neighborhood of $%
a_{n}.$ Indeed we have

\begin{lemma}
\label{Lemma magnitude of T} Assume that (A1) holds$.$ For any sequence $%
c_{n}$ such that (C1) holds$,$%
\begin{equation*}
P\left( a_{n}\leq \mathbf{T\leq }a_{n}+c_{n}\right) =1+O\left( \exp
-na_{n}c_{n}\right) .
\end{equation*}

\begin{proof}
Use Lemma \ref{Lemma Jensen}.
\end{proof}
\end{lemma}

Moreover $\mathbf{T}$ is asymptotically exponentially distributed. The
asymptotic distribution of $\mathbf{T}$ is captured in the following

\begin{lemma}
\label{Lemma approx exponential for Tbold}When (A1) holds then for all $u$\
in $\mathbb{R}^{+}$ the r.v. $\mathbf{Z:=}nt^{a_{n}}\left( \mathbf{T-}%
a_{n}\right) $ satisfies%
\begin{equation*}
p_{\mathbf{Z}}\left( u\right) =e^{-u}\left( 1+o(1)\right)
\end{equation*}%
where $m(t^{a_{n}})=a_{n}$ and therefore $\mathbf{T}=a_{n}+O_{P}\left( \frac{%
1}{na_{n}}\right) .$
\end{lemma}

\begin{proof}
Write
\begin{equation*}
p_{\mathbf{Z}}\left( u\right) =\frac{1}{nt^{a_{n}}}\frac{p_{\mathbf{S}%
_{n}/n}\left( a_{n}+u/\left( nt^{a_{n}}\right) \right) }{P\left( \mathbf{S}%
_{n}/n>a_{n}\right) }
\end{equation*}%
and use Lemmas \ref{Lemma Richter local}and \ref{Lemma Jensen}. A first
order expansion yields $a_{n}=m(t^{a_{n}})=t^{a_{n}}\left( 1+o(1)\right) $
which proves the claim.
\end{proof}

In this Section Proposition \ref{Prop approx local cond density} is extended
in order to provide an approximation of $\mathfrak{p}_{n}(Y_{1}^{k})$ when $%
Y_{1}^{k}$ is a random vector generated under $\mathfrak{p}_{n}.$ This is
obtained through an integration w.r.t. $\sigma $ in (\ref{local approx under
exact value}); indeed it holds

\begin{equation}
\mathfrak{p}_{n}(Y_{1}^{k}):=\int_{a_{n}}^{\infty }p\left( \mathbf{X}%
_{1}^{k}=Y_{1}^{k}/\mathbf{T}=\sigma \right) p_{\mathbf{T}}\left( \sigma
\right) d\sigma  \label{p_n fraktur}
\end{equation}%
and the domain of integration can be reduced to a small neighborhood of $%
a_{n}$ which contains nearly all the realizations of $\mathbf{T}$ under $%
\mathcal{E}_{n}$ . This argument allows the interchange of asymptotic
equivalents and integration.

Define
\begin{equation*}
g_{n}(x_{1}^{k}):=\int_{a_{n}}^{\infty }g_{\sigma }(x_{1}^{k})p_{\mathbf{T}%
}\left( \sigma \right) d\sigma
\end{equation*}%
where $g_{\sigma }(x)$ is defined in (\ref{g_s}).

When $g_{\sigma }$ is substituted by $g_{n}$ then
\begin{equation*}
\mathfrak{p}_{n}(Y_{1}^{k})=g_{n}(Y_{1}^{k})\left( 1+o_{\mathfrak{P}%
_{n}}(1)\right)
\end{equation*}%
does not stand.

Let
\begin{equation*}
b_{n}=a_{n}+c_{n}
\end{equation*}%
\ where $c_{n}$ is fitted compatibly with Proposition \ref{Prop approx local
cond density}.

Define
\begin{equation}
\overline{g_{n}}(y_{1}^{k}):=\frac{\int_{a_{n}}^{b_{n}}g_{\sigma
}(y_{1}^{k})p_{\mathbf{T}}\left( \sigma \right) d\sigma }{P\left( \mathbf{%
T\in }\left( a_{n},b_{n}\right) \right) }.  \label{g global}
\end{equation}

\begin{proposition}
\label{Prop p_n equiv g under S_n>na_n} When $Y_{1}^{k}$ is a random vector
generated with density $\mathfrak{p}_{n}$ and $\ $(A) and (C) hold then
\begin{equation}
\mathfrak{p}_{n}(Y_{1}^{k})=\overline{g_{n}}(Y_{1}^{k})\left( 1+o_{\mathfrak{%
P}_{n}}(a_{n}\left( \log n\right) ^{2+\delta })\right) .
\label{p_n equiv g for IS}
\end{equation}
\end{proposition}

The proof of Proposition \ref{Prop p_n equiv g under S_n>na_n} relies upon
the following Lemma, whose proof is postponed to the Appendix.

\begin{lemma}
\label{Lemma from local cond to global cond} Let $b_{n\text{ }}$ satisfy $%
b_{n}=a_{n}+c_{n}$ and (A) and (C) hold then when $Y_{1}^{n}$ is generated
under $\mathfrak{p}_{n}$ it holds%
\begin{equation*}
\mathfrak{p}_{n}\left( Y_{1}^{k}\right) =\int_{a_{n}}^{b_{n}}p\left(
Y_{1}^{k}/\mathbf{T}=\sigma \right) p\left( \mathbf{T}=\sigma \right)
d\sigma \left( 1+O_{\mathfrak{P}_{n}}\left( \exp -na_{n}c_{n}\right) \right)
.
\end{equation*}
\end{lemma}

We now prove Proposition \ref{Prop p_n equiv g under S_n>na_n} through an
integration of the local approximation given in Proposition \ref{Prop approx
local cond density}.

For all $\sigma $ in $(a_{n},b_{n})$%
\begin{equation}
g_{\sigma }\left( Y_{1}^{k}\right) =p\left( Y_{1}^{k}/\mathbf{T}=\sigma
\right) \left( 1+o_{\mathfrak{P}_{n}}(a_{n}\left( \log n\right) ^{2+\delta
})\right)  \label{g_s equiv p_s}
\end{equation}%
uniformly on $\sigma $ when $Y_{1}^{n}$ is sampled under $\mathfrak{p}_{n}.$
It then holds
\begin{eqnarray*}
\overline{g_{n}}(Y_{1}^{k}) &:&=\frac{\int_{a_{n}}^{b_{n}}g_{\sigma }\left(
Y_{1}^{k}\right) p(\mathbf{T}\mathbf{=}\sigma )d\sigma }{P\left( a_{n}<%
\mathbf{T}<b_{n}\right) } \\
&=&\int_{a_{n}}^{b_{n}}p\left( Y_{1}^{k}/\mathbf{T}=\sigma \right) p(\mathbf{%
T=\sigma })d\sigma \left( 1+o_{\mathfrak{P}_{n}}(a_{n}\left( \log n\right)
^{2+\delta })\right) \\
&=&\mathfrak{p}_{n}\left( Y_{1}^{k}\right) \left( 1+o_{\mathfrak{P}%
_{n}}(a_{n}\left( \log n\right) ^{2+\delta })\right)
\end{eqnarray*}%
where we used Lemmas \ref{Lemma from local cond to global cond} together
with (C4) which helps to keep the $o_{\mathfrak{P}_{n}}\left( a_{n}\left(
\log n\right) ^{2+\delta }\right) $ term. This concludes the proof of
Proposition \ref{Prop p_n equiv g under S_n>na_n}.

\bigskip

As a consequence of Lemma \ref{Lemmafrompntogn} the following
result holds, which asseses that when sampled under $\overline{g_{n}}$ $\ $%
the likelihood of the random vector $X_{1}^{k}$ approximates $\mathfrak{p}%
_{n}(X_{1}^{k}).$

\begin{proposition}
\label{Prop Global approx inverse} Assume (A) and (C). Let $X_{1}^{k}$ be a
random vector with p.m.\ $\overline{G_{n}}$ with density $\overline{g_{n}}$
on $\mathbb{R}^{k}$ defined in (\ref{g global}). It holds
\begin{equation*}
\overline{g_{n}}\left( X_{1}^{k}\right) =\mathfrak{p}_{n}(X_{1}^{k})\left(
1+o_{\overline{G_{n}}}(a_{n}\left( \log n\right) ^{2+\delta })\right)
\end{equation*}%
as $n\rightarrow \infty .$
\end{proposition}

\bigskip

\section{The Adaptive Twisted Importance Sampling scheme}

The last result in Proposition \ref{Prop Global approx inverse} above
suggests that an Importance Sampling density deduced from $\overline{g_{n}}$
would benefit from some optimality as defined in the Introduction since it
fits with the conditional density on long runs. It is enough to approximate
the conditional distribution of $\mathbf{T}=\mathbf{S}_{1}^{n}/n$ under $%
\mathcal{E}_{n}$ by Lemma \ref{Lemma approx exponential for Tbold} and to
plug in this approximation in (\ref{g global}).

Let $\mathbf{E}$ denote a r.v. with exponential distribution with parameter $%
na_{n}$ on $\left( a_{n},+\infty \right) $%
\begin{equation}
p_{\mathbf{E}}(t):=na_{n}e^{-na_{n}\left( t-a_{n}\right) }\mathbf{1}_{\left(
a_{n},+\infty \right) }(t).  \label{densityT}
\end{equation}

Using again Lemmas \ref{Lemma Richter local} and \ref{Lemma Jensen} it is
easily checked that
\begin{equation*}
\sup_{na_{n}\leq s\leq nb_{n}}\frac{p(\mathbf{E}=s)}{p(\mathbf{T}=s)}%
=1+o(\varepsilon _{n})
\end{equation*}%
for some sequence $\varepsilon _{n}$ whinch tends to $0,$ from which
\begin{equation}
\mathbf{g}\left( X_{1}^{k}\right) :=\frac{\int_{a_{n}}^{b_{n}}g_{\sigma
}\left( X_{1}^{k}\right) p(\mathbf{E}=\sigma )d\sigma }{%
\int_{a_{n}}^{b_{n}}p(\mathbf{E}=\sigma )d\sigma }=\mathfrak{p}_{n}\left(
X_{1}^{k}\right) \left( 1+o_{\mathfrak{P}_{n}}(\varepsilon _{n}^{\prime
})\right)  \label{g(X_1^k)}
\end{equation}%
with $\lim_{n\rightarrow \infty }\varepsilon _{n}^{\prime }=0$, which proves
that we may substitute $\mathbf{T}$ by the exponential r.v. $\mathbf{E}$
while keeping the properties of the IS procedure. We denote
\begin{equation}
\mathbf{g}\left( x_{1}^{n}\right) :=\mathbf{g}\left( x_{1}^{k}\right)
\prod_{i=k+1}^{n}\pi ^{\alpha _{k}}(x_{i})  \label{new sampling}
\end{equation}%
the sampling scheme under which the estimate (\ref{FORM IS}) is computed; in
(\ref{new sampling}) the value of $\alpha _{k}$ is defined through%
\begin{equation*}
\alpha _{k}:=m(t_{k})
\end{equation*}%
with
\begin{equation*}
m(t_{k})=\frac{n}{n-k}\left( a_{n}-\frac{s_{1}^{k}}{n}\right) .
\end{equation*}

\subsection{The Adaptive Twisted IS algorithm}

Since the r.v. $\mathbf{E}$ is highly concentrated in a small neighborhood
of $a_{n}$ we suggest to forget about $b_{n}$ in the definition (\ref{new
sampling}) of $\mathbf{g}$ above and to integrate on $(a_{n},\infty )$
instead of $(a_{n},b_{n}).$ Numerical experiments argue in favor of this
heuristic. The remarks at the end of this paragraph provide simple and
efficient solutions for the effective calculation of the estimate.

1- \ \ \ \ \ \ \ \ \ Draw $M$ independent random variables $E^{1}...,E^{M}$
with distribution (\ref{densityT}) and define the density on $\mathbb{R}^{n}$%
\begin{equation}
\overline{\mathbf{g}}(x_{1}^{n}):=\frac{1}{M}\sum_{m=1}^{M}\left(
g_{E^{m}}(x_{1}^{k})\prod_{i=k+1}^{n}\pi ^{\alpha
_{k}}(x_{i}^{i})\right)  \label{gbar}
\end{equation}%
where $g_{E^{m}}$ is defined as
\begin{equation}
g_{E^{m}}(x_{1}^{k}):=\prod_{i=1}^{k-1}g_{i+1}(x_{i+1}/x_{1}^{i})
\label{g_T}
\end{equation}%
where $g_{i+1}(x_{i+1}/x_{1}^{i})$ is defined in (\ref{g_i}) for $i\geq 1$ ,
$g_{0}(x_{1})=\pi ^{a_{n}}(x_{1})$ and
\begin{equation}
\pi ^{\alpha _{k}}(x):=\frac{\exp t_{k}x}{\Phi (t_{k})}p(x)
\label{pi^alpha_k}
\end{equation}%
where $\alpha _{k}=m(t_{k})$ and $t_{k}$ $\ $is the only solution of the
equation
\begin{equation}
m(t_{k})=\frac{n}{n-k}\left( a_{n}-\frac{s_{1}^{k}}{n}\right)
\label{start remaining tilted n-k rv's}
\end{equation}%
\ with $s_{1}^{k}:=x_{1}+...+x_{k},$ with $s_{1}^{0}=0.$

2-Define $L$ which is the number of replications of the simulated random
trajectory to be performed

3-For $l$ between $1$ and $L$ do

\ \ \ \ \ \ \ \ \ \ \{

\ \ \ \ \ \ \ \ \ draw a random variable $E(l)$ with distribution (\ref%
{densityT})

\ \ \ \ \ \ \ \ \ \ draw the first $k$ variables $X_{1}^{k}(l)$ recursively
with density $g_{E(l)}(x_{1}^{k})$ as defined in (\ref{g_T}) with $E^{m}$
substituted by $E(l)$.

\ \ \ \ \ \ \ \ \ \ Draw the $n-k$ random variables $X_{k+1}^{n}(l)$
independently with common density $\pi ^{\alpha _{k}}(x)$ defined in (\ref%
{pi^alpha_k}) with $E^{m}$ substituted by $E(l).$

\}

4- Define

\ \ \ \ \ \ \ \ \
\begin{equation}
\widehat{P_{n}}:=\frac{1}{L}\sum_{l=1}^{L}\frac{\prod_{i=0}^{n}p(X_{i}(l))}{\overline{g}(X_{1}^{n}(l))}
\mathbf{1}_{\mathcal{E}_{n}}(l)  \label{Algo BR}
\end{equation}

where
\begin{equation}
\mathbf{1}_{\mathcal{E}_{n}}(l):=\mathbf{1}_{(a_{n},\infty )}\left(
S_{1}^{n}(l)/n\right)  \label{IndicS_n/n>a_n}
\end{equation}

\subsubsection{Some remarks for the implementation of the algorithm}

A number of remarks hereunder show that ATIS is not difficult to implement.
Since the first order efficiency of i.i.d sample schemes is reached if and
only if the sampling distribution is the twisted one with parameter $a_{n}$
(see \cite{Ermakov2007}), the present algorithm should be compared with it.
The classical IS scheme which uses i.i.d. replicates with density $\pi
^{a_{n}}$ is easy to implement but may lead to biased estimates of $P_{n}$;
the simulation of a r.v. with density $\pi ^{a_{n}}$ is difficult in non
standard cases When $p$ is easy to simulate then an Acceptance/Rejection
algorithm can be used; however this requires to truncate the support of $p$,
what should precisely be avoided in order to obtain unbiaised estimates; see
\cite{BarbeBroniatowski1999}. When $\pi ^{a_{n}}$ is easy to simulate, ATIS
may take more time to run, due to the various intermediate calculations
which are required at each stage of the algorithm.

\

The generation of the r.v. $X_{1}^{k}(l)$ above is easy and fast and does
not require any simulation according to a twisted density. It holds
\begin{equation}
g_{i+1}(x_{i+1}/x_{1}^{i})=C_{i}p(x_{i+1})\mathfrak{n}\left(
ab,a,x_{i+1}\right)  \label{g_i+1(x_i+1/x_i)numerically}
\end{equation}%
where $\mathfrak{n}\left( \mu ,\sigma ^{2},x\right) $ is the normal density
with mean $\mu $ and variance $\sigma ^{2}$ at $x$. Here
\begin{equation*}
a=s_{i,n}^{2}\left( n-i-1\right)
\end{equation*}%
\begin{equation*}
b=t_{i,n}+\frac{\mu _{3}^{(i,n)}}{2s_{i,n}^{4}\left( n-i-1\right) }.
\end{equation*}%
A r.v. $Y$ with density $g(x)=Cp(x)n(x),$ with $C=\left( \int
p(x)n(x)dx\right) ^{-1}$ and where $p$ is a given density and $n(x)=%
\mathfrak{n}\left( \mu ,\sigma ^{2},x\right) $ is easy to simulate: Denote $%
\mathfrak{N}$ the c.d.f. with density $\mathfrak{n}\left( \mu ,\sigma
^{2},x\right) .$ It is easily checked that $g(x)$ is the density of the r.v.
$Y:=\mathfrak{N}^{\leftarrow }(X)$ where $X$ is a r.v.\ on $\left[ 0,1\right]
$ with density $h(u):=\frac{1}{C}p\left( \mathfrak{N}^{\leftarrow
}(u)\right) ;\mathfrak{N}^{\leftarrow }$ denotes the reciprocal function of $%
\mathfrak{N}$ $.$ Now an acceptance/rejection algorithm provides a
realisation of $X$. Indeed let $f(x)$ be a density such that $p\left(
\mathfrak{N}^{\leftarrow }(u)\right) \leq Kf(x)$ for some constant $K$ and
all $x$ in $\left[ 0,1\right] $; Let $\mathcal{P}$ be uniformly distributed
on the hypograph of $Kf,$ namely $\mathcal{P}:=(X_{\mathcal{P}},Y_{\mathcal{P%
}}=KUf\left( X_{\mathcal{P}}\right) )$ where $X_{\mathcal{P}}$ has density $%
f $ and $U$ is uniform $\left[ 0,1\right] $ independent of $X_{\mathcal{P}}.$
When $Y_{\mathcal{P}}$ is less than $p\left( \mathfrak{N}^{\leftarrow }(X_{%
\mathcal{P}})\right) $ then $X_{\mathcal{P}}$ has density $h.$

The calculation of $\overline{\boldsymbol{g}}(X_{1}^{n}(l))$ above requires
the value of $C_{i}=\left( \int p(x)\mathfrak{n}\left( ab,a,x\right)
dx\right) ^{-1}$ in (\ref{g_i+1(x_i+1/x_i)numerically}). A Monte Carlo
technique can be used: simulate $N$ i.i.d.\ r.v's $Z_{j}$ with density $%
\mathfrak{n}\left( ab,a,.\right) $, which is fast, and substitute $C_{i}$ by
$\widehat{C_{i}}:=\left( \frac{1}{N}\sum_{j=1}^{N}p(Z_{j})\right) ^{-1}$,
which provides a very accurate approximation to be inserted in the
calculation of the estimate.

It may seem that this algorithm requires to solve $Lk$ equations of the form
$m(t)=\frac{n}{n-i}\left( \sigma -\frac{S_{1}^{i}}{n}\right) $ in order to
obtain the $t_{i,n}$ which are necessary to perform the simulation of $%
X_{1}^{k}(l)$ as described above as well as the calculation of $\overline{%
\mathbf{g}}(x_{1}^{n})$.\ Such is is not the case, and only $L$ equations
have to be solved. Consider for example the simulation of $X_{1}^{k}(l)$
with density $g_{E(l)}(x_{1}^{k}).$ This is achieved as follows:

1- Solve the equation
\begin{equation*}
m(t)=E(l)
\end{equation*}%
whose solution is $t_{0,n}.$ Generate $X_{0}(l)$ according to $\pi ^{E(l)}.$

2- Since
\begin{equation*}
m(t_{i+1,n})-m(t_{i,n})=-\frac{1}{n-i}\left( m(t_{i,n})+X_{i}(l)\right)
\end{equation*}%
use a first order approximation to derive%
\begin{equation*}
t_{i+1,n}\simeq t_{i,n}-\frac{1}{\left( n-i\right) s(t_{i,n})}\left(
m(t_{i,n})+X_{i}(l)\right)
\end{equation*}%
from which (\ref{g_i+1(x_i+1/x_i)numerically}) is derived and $X_{i+1}(l)$
can be simulated as mentioned above.\ In the moderate deviation scale the
function $s^{2}(.)$ does not vary from $1$ and the above approximation is
fair.

\begin{remark}
The density $\overline{\mathbf{g}}(x)$ on $\mathbb{R}^{n}$ is a Monte Carlo
approximation of $g_{n}$ defined by%
\begin{equation*}
g_{n}(x):=\int g_{\sigma }(x)p\left( \mathbf{T}=\sigma \right) d\sigma
\end{equation*}%
where $p\left( \mathbf{T}=\sigma \right) $ is replaced by $p(\mathbf{E=}%
\sigma )$ and the integral is replaced by a finite mixture.\ $M$ is a free
parameter. Also notice that the $n-k$ i.i.d.\ r.v's have common tilted
density $\pi ^{\alpha _{k}}(x)$ with parameter given by (\ref{start
remaining tilted n-k rv's}), thus identical to Ermakov's sampling scheme
with end point in$\left( a_{n}-\frac{s_{1}^{k}}{n},\infty \right) ,$ and not
in $\left( m(t_{k-1})-\frac{s_{1}^{k}}{n},\infty \right) .$\bigskip
\end{remark}

\subsection{The choice of the tuning parameters}

\subsubsection{Choosing $k$}

The critical parameter $k$ is the length of the partial sum run which is to
be simulated according to the density $\mathbf{g}(x_{1}^{k})$ as defined in (%
\ref{g(X_1^k)}). By (\ref{g(X_1^k)}) it would be enough to establish some
statistics averaging the estimate ratios $\mathbf{g}\left( X_{1}^{j}\right) /%
\mathfrak{p}_{n}\left( X_{1}^{j}\right) $ on a set of runs , and to select $%
k $ as some $j$ ensuring that this ratio keeps close to $1.$ In the case
when the r.v's $X_{i}$ are normally distributed the density $%
g_{i}(y_{i+1}/y_{1}^{i})$ as defined in (\ref{g_i}) coincides with $%
p(y_{i+1}/y_{i},\frac{S_{n}}{n}=\sigma )$ for all value of $i$ between $1$
and $n-1$ which entails that $k$ can be set equal to $n-1.$ This very
peculiar case is illustrated in Figure 1, for $n=100,$ and $P_{n}$ is close
to $0.01$. We can see that ATIS\ produces a very sharp estimate of $P_{n}$
for a small value of $L$ when compared to the classical IS scheme.

In the other cases, when $g_{i}(y_{i+1}/y_{1}^{i})$ approximates $%
p(y_{i+1}/y_{i},\frac{S_{n}}{n}=\sigma )$ only under some conditions on $k$
as described in Conditions (A), we propose the following heuristics, which
works well and is easy to implement; other choices are possible, which
provide similar acceptable results. Instead of $\mathbf{g}$ consider the
following construction, which will also be used in the IS algorithm:
simulate $E^{1}...,E^{M}$ , i.i.d. with distribution (\ref{densityT}) and
define the density on $\mathbb{R}^{j+1}$%
\begin{equation*}
\overline{\mathbf{g}}(x_{0}^{j}):=\frac{1}{M}%
\sum_{m=1}^{M}g_{E^{m}}(x_{1}^{j})\pi ^{E^{m}}(x_{0})
\end{equation*}%
where $g_{E^{m}}$ is defined as
\begin{equation*}
g_{E^{m}}(x_{0}^{j}):=\prod_{i=1}^{j-1}g_{i+1}(x_{i+1}/x_{1}^{i})
\end{equation*}%
where $g_{i+1}(x_{i+1}/x_{1}^{i})$ is defined in (\ref{g_i}) for $i\geq 1$.
The density $\overline{\mathbf{g}}(x_{0}^{j})$ is a Monte Carlo
approximation of $\mathbf{g}\left( x_{0}^{j}\right) $.

By (\ref{p_n fraktur}) and following the same heuristics as for $\overline{%
\mathbf{g}}$ define , with a new set of i.i.d. $E^{m}$'s
\begin{equation*}
\overline{\mathfrak{p}_{n}}\left( x_{0}^{j}\right) :=\frac{1}{M}%
\sum_{m=1}^{M}p\left( \mathbf{X}_{0}^{j}=x_{0}^{j}/\mathbf{T}=E^{m}\right) .
\end{equation*}%
We use Lemma \ref{Lemma Richter local} in order to obtain an explicit
approximation for $\overline{\mathfrak{p}_{n}}.$ It holds
\begin{eqnarray*}
p\left( \mathbf{X}_{0}^{j}=x_{0}^{j}/\mathbf{T}=E^{m}\right) &=&\frac{%
p\left( \mathbf{S}_{j+1}^{n}=n\left( E^{m}-\frac{s_{0}^{j}}{n}\right)
\right) }{p\left( \mathbf{S}_{1}^{n}=nE^{m}\right) }p\left( \mathbf{X}%
_{0}^{j}=x_{0}^{j}\right) \\
&=&\sqrt{\frac{n-j}{n}}\frac{\exp -(n-j)I\left( \frac{n}{n-j}\left( E^{m}-%
\frac{s_{0}^{j}}{n}\right) \right) }{\exp -nI\left( E^{m}\right) }p\left(
\mathbf{X}_{0}^{j}=x_{0}^{j}\right) \left( 1+o(1)\right) .
\end{eqnarray*}%
Define therefore
\begin{equation*}
\left( \widehat{\mathfrak{p}_{n}}\right) _{m}\left( x_{0}^{j}\right) :=\sqrt{%
\frac{n-j}{n}}\frac{\exp -(n-j)I\left( \frac{n}{n-j}\left( E^{m}-\frac{%
s_{0}^{j}}{n}\right) \right) }{\exp -nI\left( E^{m}\right) }p\left( \mathbf{X%
}_{0}^{j}=x_{0}^{j}\right)
\end{equation*}%
and
\begin{equation*}
\widehat{\mathfrak{p}_{n}}\left( x_{0}^{j}\right) :=\frac{1}{M}%
\sum_{m=1}^{M}\left( \widehat{\mathfrak{p}_{n}}\right) _{m}\left(
x_{0}^{j}\right) .
\end{equation*}%
Fix some integer $L$ which is the number of runs to be simulated in order to
fix $k;$ $L$ need not be large. For all $l$ between $0$ and $L$ draw
independently a random variable $E^{l}$ with density (\ref{densityT}) and
the run $X_{0}^{j}(l)$ with density $g_{E^{l}}$ defined as in (\ref{g_s})
with $k$ substituted by $j$ and $\sigma $ by $E^{l}.$
\begin{equation*}
\frac{1}{L}\sum_{l=1}^{L}\frac{\overline{\mathbf{g}}(X_{0}^{j}(l))}{\widehat{%
\mathfrak{p}_{n}}\left( X_{0}^{j}(l)\right) }.
\end{equation*}%
Fix $k$ as the smallest $j$ which indicates a departure of this statistics
from $1.$

\subsubsection{The choice of $M$}

In ATIS\ the distribution in (\ref{g(X_1^k)}) is substituted by a numerical
approximation of
\begin{equation}
\int_{a_{n}}^{\infty }g_{\sigma }\left( X_{1}^{k}\right) p(\mathbf{E}=\sigma
)d\sigma  \label{gintegralE}
\end{equation}%
which is suboptimal with respect to (\ref{g(X_1^k)}) but is easily
implemented.\ A Monte Carlo procedure produces $\overline{\mathbf{g}}%
(x_{1}^{n})$ as described above in (\ref{gbar}). It appears that $M$ should
be large when $k$ is large. For example in the normal case with $n=100,$ for
$k=60$, then $M=30$ produces excellent estimates for values of $L$ of order $%
5000,$ whereas for $k=98,$ the value of $M$ should increase up to $2000,$
with the same $L$ .as seen in Figure2$.$ The reason for this increase in $M$
is that (\ref{gintegralE}) is a mixture of densities in very high dimension,
\ which seems very sensitive with respect to the approximation of the
mixture measure. This point should deserve a specific study, out of the
scope of the present paper.\ However the normal case is quite specific,
since it allows $k$ to be as close to $n$ as wanted. In the other cases, as
examplified in the figures pertaining to the exponential case, $k$ is
resticted to lower values and $M$ is rather low.

\subsection{Asymptotic efficiency of the adaptive twisted IS scheme}

The evaluation of the performances of IS algorithm is a controversal
argument.\ Many criterions are at hand, for example the \textit{probability
of hits }which counts the relative number of simulations hitting the target $%
\left( a_{n},\infty \right) ,$ \ or the \textit{variance} of the estimator.
We refer to the book by Bucklew \cite{Bucklew2004} for a discussion on the
relative merits of each approach.\

\bigskip The variance of an IS estimate of $P_{n}$ under the sampling
density $g$ writes%
\begin{equation*}
VarP_{g}^{(n)}(\mathcal{E})=\frac{1}{L}\left( E_{g}\left( P_{g}(l)\right)
^{2}-P_{n}^{2}\right)
\end{equation*}%
with
\begin{equation*}
P_{g}(l):=\frac{p\left( Y_{1}^{n}\left( l\right) \right) }{g\left(
Y_{1}^{n}\left( l\right) \right) }\mathbf{1}_{\mathcal{E}_{n}}\left( \Sigma
_{1}^{n}\left( l\right) \right) .
\end{equation*}

The situation which we face with our proposal lacks the possibility to
provide an order of magnitude of the variance our our IS estimate, since the
properties necessary to define it have been obtained only on \textit{typical
paths} under the sampling density $\mathbf{g}$ \ defined in (\ref{new
sampling}) and not on the whole space $\mathbb{R}^{n}$ (but in the case when
the $X_{i}$'s are normally distributed)$.$ We will prove , however, that the
performance of this new procedure can be considered favorably. Not
surprisingly the loss of performance with respect to the optimal sampling
density $p_{\mathbf{X}_{1}^{n}/\mathcal{E}_{n}}$ is due to the $n-k$ last
i.i.d. simulations, leading a quasi- MSE of the estimate proportional to $%
\sqrt{n-k}.$

In order to discuss this we first go back to the classical IS scheme, for
which we evaluate the asymptotic variance.

\subsubsection{The variance of the classical IS scheme and a discussion on
efficiency}

The asymptotic variance of the estimate of $P(\mathcal{E}_{n})$ can be
evaluated as follows.

The classical IS is defined simulating $L$ times a random sample of $n$
i.i.d. r.v's $X_{1}^{n}(j)$, $1\leq j\leq L,$ with tilted density $\pi
^{a_{n}}$. The standard IS estimate is defined through%
\begin{equation*}
\overline{P_{n}}:=\frac{1}{L}\sum_{l=1}^{L}\mathbf{1}_{\mathcal{E}_{n}}(l)%
\frac{\prod_{i=1}^{n}p(X_{i}(l))}{\prod_{i=1}^{n}\pi
^{a_{n}}(X_{i}(l))}
\end{equation*}%
where the $X_{i}(l)$ are i.i.d. with density $\pi ^{a_{n}}$ and $\mathbf{1}_{%
\mathcal{E}_{n}}(l)$ is as in (\ref{IndicS_n/n>a_n})$.$ Set
\begin{equation*}
\overline{P_{n}}(l):=\mathbf{1}_{\mathcal{E}_{n}}(l)\frac{%
\prod_{i=1}^{n}p(X_{i}(l))}{\prod_{i=1}^{n}\pi
^{a_{n}}(X_{i}(l))}.
\end{equation*}%
The variance of $\overline{P_{n}}$ is given by
\begin{equation*}
Var\overline{P_{n}}=\frac{1}{L}\left( E_{\pi ^{a_{n}}}\left( \overline{P_{n}}%
(l)\right) ^{2}-P_{n}^{2}\right) .
\end{equation*}%
The \textit{relative accuracy }of the estimate $P_{n}^{IS}$ is defined
through
\begin{equation*}
RE(\overline{P_{n}}):=\frac{Var\overline{P_{n}}}{P_{n}^{2}}=\frac{1}{L}%
\left( \frac{E_{\pi ^{a_{n}}}\left( \overline{P_{n}}(l)\right) ^{2}}{%
P_{n}^{2}}-1\right) .
\end{equation*}%
It holds

\begin{proposition}
\label{Prop rel efficiency standard IS}The relative accuracy of the estimate
$P_{n}^{IS}$ is given by
\begin{equation*}
RE(\overline{P_{n}})=\frac{\sqrt{2\pi }\sqrt{n}}{L}a_{n}(1+o(1))\text{ as }n%
\text{ tends to infinity.}
\end{equation*}
\end{proposition}

\begin{proof}
It holds, omitting the index $l$ for brevity and noting $a$ for $a_{n}$
\begin{eqnarray*}
E_{\pi ^{a}}\left( \overline{P_{n}}(l)\right) ^{2} &=&E_{p}\left( \mathbf{1}%
_{\mathcal{E}_{n}}(X_{1}^{n})\frac{p(X_{1}^{n})}{\pi ^{a}(X_{1}^{n})}\right)
\\
&=&\phi ^{n}(t^{a})\exp -nat^{a}\int_{na}^{\infty }\exp -t^{a}\left(
s-na\right) p_{S_{n}}(s)ds.
\end{eqnarray*}%
The Laplace integral above satisfies%
\begin{equation*}
\int_{na}^{\infty }\exp -t^{a}\left( s-na\right) p_{S_{n}}(s)ds=P_{n}(1+o(1))
\end{equation*}%
as $n$ tends to infinity, which, together with the expansion
\begin{equation*}
\phi ^{n}(t_{a})\exp -nat^{a}=P_{n}\sqrt{n}\sqrt{2\pi }t^{a}(1+o(1))
\end{equation*}%
(which holds when $\lim_{n\rightarrow \infty }a\sqrt{n}$ $=\infty $)
concludes the proof. We have used Lemma \ref{Lemmaorderoftinetc} to
assess that $\lim_{n\rightarrow \infty }s(t^{a})=1.$
\end{proof}

We now come to a discussion of the above result.\ It is well known that the
variance is not a satisfactory criterion to describe the variability of the
outcomes of a random phenomenon: for example, a sequence of symmetric r.v's $%
\mathbf{X}_{n}$ taking values $-\exp \exp n,0,\exp \exp n$ with relative
frequencies defined through $P(\mathbf{X}_{n}=\exp n)=\exp -n$ has variance
going to $\infty $ while being concentrated at $0.$ In this case we can
define an increasing family of sets $B_{n}$ with $P(\mathbf{X}_{n}\in
B_{n})\rightarrow 1$ on which $E\left( \mathbf{1}_{B_{n}}\mathbf{X}%
_{n}^{2}\right) =0,$ a much better indicator, obtained through trimming. We
will prove that such an indicator cannot be defined for the classical IS
scheme, stating therefore that the variance rate obtained in Proposition \ref%
{Prop rel efficiency standard IS} is indeed meaningfull.

The easy case when $\mathbf{X}_{1},...,\mathbf{X}_{n}$ are i.i.d. with
standard normal distribution is sufficient for our need.

The variance of the IS estimate is proportional to
\begin{eqnarray*}
V &:&=E_{p}\mathbf{1}_{\left( na_{n},\infty \right) }\left( \mathbf{S}%
_{1}^{n}\right) \frac{p\left( \mathbf{X}_{1}^{n}\right) }{\pi ^{a_{n}}\left(
\mathbf{X}_{1}^{n}\right) }-P_{n}^{2} \\
&=&E_{p}\mathbf{1}_{\left( na_{n},\infty \right) }\left( \mathbf{S}%
_{1}^{n}\right) \left( \exp \frac{na_{n}^{2}}{2}\right) \left( \exp -a_{n}%
\mathbf{S}_{1}^{n}\right) -P_{n}^{2}
\end{eqnarray*}%
A set $B_{n}$ resulting as reducing the MSE should penalize large values of $%
-\mathbf{S}_{1}^{n}$ while bearing nearly all the realizations of $\mathbf{S}%
_{1}^{n}$ under the i.i.d. sampling scheme $\pi ^{a_{n}}$ as $n$ tends to
infinity. It should therefore be of the form $\left( nb_{n},\infty \right) $
for some $b_{n}$ so that

(a)
\begin{equation*}
\lim_{n\rightarrow \infty }E_{\pi ^{a_{n}}}\mathbf{1}_{\left( nb_{n},\infty
\right) }\left( \mathbf{S}_{1}^{n}\right) =1
\end{equation*}%
and

(b)%
\begin{equation*}
\lim_{n\rightarrow \infty }\sup \frac{E_{p}\mathbf{1}_{\left( na_{n},\infty
\right) \cap \left( nb_{n},\infty \right) }\left( \mathbf{S}_{1}^{n}\right)
\frac{p\left( \mathbf{X}_{1}^{n}\right) }{\pi ^{a_{n}}\left( \mathbf{X}%
_{1}^{n}\right) }}{V}<1
\end{equation*}%
which means that the IS\ sampling density $\pi ^{a_{n}}$ can lead a MSE
defined by
\begin{equation*}
MSE(B_{n}):=E_{p}\mathbf{1}_{\left( na_{n},\infty \right) \cap \left(
nb_{n},\infty \right) }\frac{p\left( \mathbf{X}_{1}^{n}\right) }{\pi
^{a_{n}}\left( \mathbf{X}_{1}^{n}\right) }-P_{n}^{2}
\end{equation*}%
with a clear gain over the variance indicator. However when $b_{n}\leq a_{n}$
(b) does not hold and when $b_{n}>a_{n}$ (a) does not hold.

So no reduction of this variance can be obtained by taking into account the
properties of the \textit{typical paths }generated under the sampling
density: a reduction of the variance is possible only by conditioning on
"small" subsets of the sample paths space. On no classes of subsets of $%
\mathbb{R}^{n}$ with probability going to $1$ under the sampling is it
possible to reduce the variability of the estimate, whose rate is definitely
proportional to $\sqrt{n},$ imposing a burden of order $L\sqrt{n}\alpha $ in
order to achieve a relative efficiency of $\alpha \%$ with respect to $%
P_{n}. $

\bigskip

\subsubsection{The MSE\ of our estimate on a growing class of typical paths}

We will evaluate the performance of our estimate under $\mathbf{g}$ since
the algorithm envolves technical parameters (typically $M$); in practice the
Monte Carlo approximation introduces no significant bias.

At the contrary to just evidenced hereabove, the procedure which we propose
has a small asymptotic variability when evaluated through trimming on
classes of subsets of $\mathbb{R}^{n}$ whose probability goes to $1$ under
the sampling $\mathbf{g}$ .\ These subsets of $\mathbb{R}^{n}$ get smaller
and smaller as $n$ increases as measured through the MSE of the estimate
with respect to the MSE of the classical IS estimate.

We prove the existence of these trimming sets in the present section and
state that the resulting gain in terms of the MSE of our estimate is the
proper measure of its performance.

These sets are the $C_{n}$ described in the following Lemma, whose proof is
differed to the appendix. For sake of notational simplicity denote $%
\varepsilon _{n}$ the $\varepsilon _{n}^{\prime }$ defined in (\ref{g(X_1^k)}%
).

\begin{lemma}
\label{Lemma set C_n for efficiency} With the just mentioned $\varepsilon
_{n},$ define the family of sets $C_{n}$in $\mathbb{R}^{n}$ such that for
all $x_{1}^{n}$ in $C_{n},$%
\begin{equation*}
\left\vert \frac{\mathfrak{p}_{n}(x_{1}^{k})}{\mathbf{g}\left(
x_{1}^{k}\right) }-1\right\vert <\varepsilon _{n}
\end{equation*}%
and
\begin{equation*}
\left\vert \frac{m(t_{k})}{a_{n}}-1\right\vert <\delta _{n}
\end{equation*}%
where $t_{k}$ \ is defined through%
\begin{equation*}
m(t_{k}):=\frac{n}{n-k}\left( a_{n}-\frac{s_{1}^{k}}{n}\right)
\end{equation*}%
and $\delta _{n}$ satisfies
\begin{equation*}
\lim_{n\rightarrow \infty }\delta _{n}=0
\end{equation*}%
together with%
\begin{equation*}
\lim_{n\rightarrow \infty }\delta _{n}a_{n}\sqrt{n-k}=\infty .
\end{equation*}%
Then
\begin{equation*}
\lim_{n\rightarrow \infty }\mathbf{G}\left( C_{n}\right) =1.
\end{equation*}%
Furthermore on $C_{n}$%
\begin{equation}
t_{k}s(t_{k})=a_{n}\left( 1+o(1)\right) .  \label{order of t_k}
\end{equation}
\end{lemma}

We now prove that our IS algorithm provides a net improvement over the
classical IS scheme in terms of Mean Square Error when evaluated on this
family of sets.

Define%
\begin{equation*}
RE\left( \widehat{P_{n}}\right) =\frac{1}{L}\left( \frac{E_{\mathbf{g}%
}\left( \mathbf{1}_{C_{n}}\widehat{P_{n}}(l)\right) ^{2}}{P_{n}^{2}}-1\right)
\end{equation*}%
\begin{equation*}
\widehat{P_{n}}(l):=\frac{\prod_{i=0}^{n}p(X_{i}(l))}{\mathbf{g}%
(X_{1}^{n}(l))}\mathbf{1}_{\mathcal{E}_{n}}\left( S_{1}^{n}(l)\right) .
\end{equation*}%
We prove that

\begin{proposition}
\label{Prop rel efficiency IS}The relative accuracy of the estimate $%
P_{n}^{IS}$ is given by
\begin{equation*}
RE(\widehat{P_{n}})=\frac{\sqrt{2\pi }\sqrt{n-k-1}}{L}a_{n}(1+o(1))\text{ as
}n\text{ tends to infinity.}
\end{equation*}
\end{proposition}

\begin{proof}
Denote $E_{\mathfrak{P}_{n}}$ the expectation with respect to the p.m. $%
\mathfrak{P}_{n}$ of $X_{1}^{n}(l)$ conditioned upon $\mathcal{E}%
_{n}(l):=\left( S_{1}^{n}(l)/n>a_{n}\right) $; we omit the index $l$ for
brevity. Using the definition of $C_{n}$ we get%
\begin{eqnarray*}
E_{\mathbf{g}}\left( \mathbf{1}_{C_{n}}\widehat{P_{n}}(l)\right) ^{2}
&=&P_{n}E_{\mathfrak{P}_{n}}\mathbf{1}_{C_{n}}(X_{1}^{n})\frac{%
p(X_{1}^{k})p(X_{k+1}^{n})}{\mathbf{g}(X_{1}^{k})\mathbf{g}%
(X_{k+1}^{n}/X_{1}^{k})} \\
&\leq &P_{n}(1+\varepsilon _{n})E_{\mathfrak{P}_{n}}\mathbf{1}%
_{C_{n}}(X_{1}^{n})\frac{p(X_{1}^{k})}{p(X_{1}^{k}/\mathcal{E}_{n})}\frac{%
p(X_{k+1}^{n})}{\mathbf{g}(X_{k+1}^{n}/X_{1}^{k})} \\
&=&P_{n}^{2}(1+\varepsilon _{n})E_{\mathfrak{P}_{n}}\mathbf{1}%
_{C_{n}}(X_{1}^{n})\frac{1}{p(\mathcal{E}_{n}/X_{1}^{k})}\frac{p(X_{k+1}^{n})%
}{\mathbf{g}(X_{k+1}^{n}/X_{1}^{k})} \\
&=&P_{n}^{2}(1+\varepsilon _{n})\sqrt{2\pi }\sqrt{n-k-1}E_{\mathfrak{P}_{n}}%
\mathbf{1}_{C_{n}}(X_{1}^{n})t_{k}s(t_{k})(1+o(1)) \\
&=&P_{n}^{2}\sqrt{2\pi }\sqrt{n-k-1}a_{n}(1+o(1)).
\end{eqnarray*}%
The second line uses $A_{\varepsilon _{n}}^{k}.$ The third line is Bayes
formula.\ The fourth line is Lemma \ref{Lemma Jensen}.\ The fifth line uses (%
\ref{order of t_k}) and uniformity in Lemma \ref{Lemma Jensen}, where the
conditions in Corollary 6.1.4\ of Jensen (1995) are easily checked since, in
his notation, $J(\theta )=\mathbb{R}$ , condition (i) holds for $\theta $ in
a neighborhood of $0$ ($\Theta _{0}$ undeed is resticted to such a set in
our case), (ii) clearly holds and (iii) is (\ref{f.c. Edgeworth}).
\end{proof}

\begin{proposition}
\label{Prop rel efficiency}When $a_{n}=n^{-\gamma }$ then under (A) the
ratio of the relative efficiencies of the Adaptive IS algorithm with respect
to the standard IS scheme is of order $\sqrt{n-k}/\sqrt{n}$.. The same
result holds when $a_{n}=\left( \log n\right) ^{-\alpha }$
\end{proposition}

\section{Importance Sampling for M-estimators}

This Section provides some application of the previous results for some
classical types of estimators for which sharp moderate deviation
probabilities can be obtained through linear approximations. We follow
closely the work by \cite{Ermakov2007}; see also \cite{Arcones2002}.

Let $T$ denote a real valued statistical functional defined on the space $%
M_{F}$, where we assume that $T$ has an Influence Function. Let $P$ be a
given p.m.\ We assume that for all \ $Q$ in $M_{F}$ there exists a function $%
g$ (depending on $P$) such that
\begin{equation}
\left\vert T(Q)-T(P)-\int gdQ\right\vert <\omega \left( N\left( Q-P\right)
\right)  \label{diff T}
\end{equation}

where $N$ is a seminorm defined on $M_{F}$ , continuous in the $\tau _{F}$
topology, and $\omega $ is a continuous and strictly monotone function which
satisfies $\omega (t)/t\rightarrow 0$ as $t\rightarrow 0.$

The function $g$ is the \textit{Influence Function }of $T$ \ at $P.$ The
class $F$ considered here contains $B(\mathbb{R)\cup }\left\{ g\right\} .$

Let $\psi (x,t)$ be real valued function defined on $\mathbb{R}^{2}$ and
assume that $P$ satisfies $\int \left\vert \psi (x,t)\right\vert dP<\infty .$
Define $T(P)$ as any solution of the equation
\begin{equation}
\int \psi (x,t)dP=0  \label{M estim}
\end{equation}%
if defined. When $P=P_{t}$ depends upon a real valued parameter $t$ such
that
\begin{equation*}
T(P_{t})=t
\end{equation*}%
then $T$ is \textit{Fisher consistent} and the substitution of $P_{t}$ by $%
P_{n}$ in (\ref{M estim})$,$ the empirical measure pertaining to an i.i.d.
sample with unknown p.m. $P_{t_{0}}$ provides a consistent estimate of $%
t_{0} $ under appropriate regularity conditions; see \cite{Sertfling1980}.
Such estimate is an M-estimator. We assume that all conditions M1 to M5 in
\cite{Ermakov2007} hold, which implies that (\ref{diff T}) above holds (see
\cite{Ermakov2007} Theorem 4.2). Also in this case the function $g$ writes%
\begin{equation*}
g(x)=\frac{\psi (x,t_{0})}{\frac{d}{dt}\int \psi (x,t_{0})dP_{t_{0}}}.
\end{equation*}%
The same situation holds for L-estimators,

When (\ref{Sanov MDP}) holds in $M_{F}$ it can be checked that a strong MDP
holds for $T(P_{n});$ Indeed when $g$ belongs to the class $F$ and
\begin{equation*}
\lim_{n\rightarrow \infty }\left( na_{n}^{2}\right) ^{-1}\log \left[
nP_{t_{0}}\left( \left\vert g(X_{1})\right\vert >na_{n}\right) \right]
=-\infty
\end{equation*}%
then using (\ref{diff T}) and (\ref{Sanov MDP}) it can be proved that the
remaining term in $T(P_{n})-T(P_{t_{0}})$ is negligible w.r.t. the linear
approximation $\int g(x,t_{0})dP_{n}$ on the moderate deviation scale, as
follows from (2.14) and (2.15) in \cite{Ermakov2007}. Furthermore in this
case the strong moderate deviation holds for $P_{t_{0}}\left( \left\vert
T(P_{n})-T\left( P_{t_{0}}\right) \right\vert >a_{n}\right) $ and
\begin{equation*}
\lim_{n\rightarrow \infty }\frac{P_{t_{0}}\left( T(P_{n})-T\left(
P_{t_{0}}\right) >a_{n}\right) }{P_{t_{0}}\left( \frac{1}{n}%
\sum_{i=1}^{n}g(X_{i})>a_{n}\right) }=1
\end{equation*}%
in the range $a_{n}=n^{-\alpha },\frac{1}{3}<\alpha <\frac{1}{2};.$ see also
Inglot, Kallenberg and Ledwina \cite{InglotKallenbergLedwina1992}.

\section{Simulation results}

\subsection{The gaussian case}

\subsubsection{\protect\bigskip Typical paths under the final value}
\begin{figure}
  \centering
  \includegraphics[width=5cm,angle=-90]{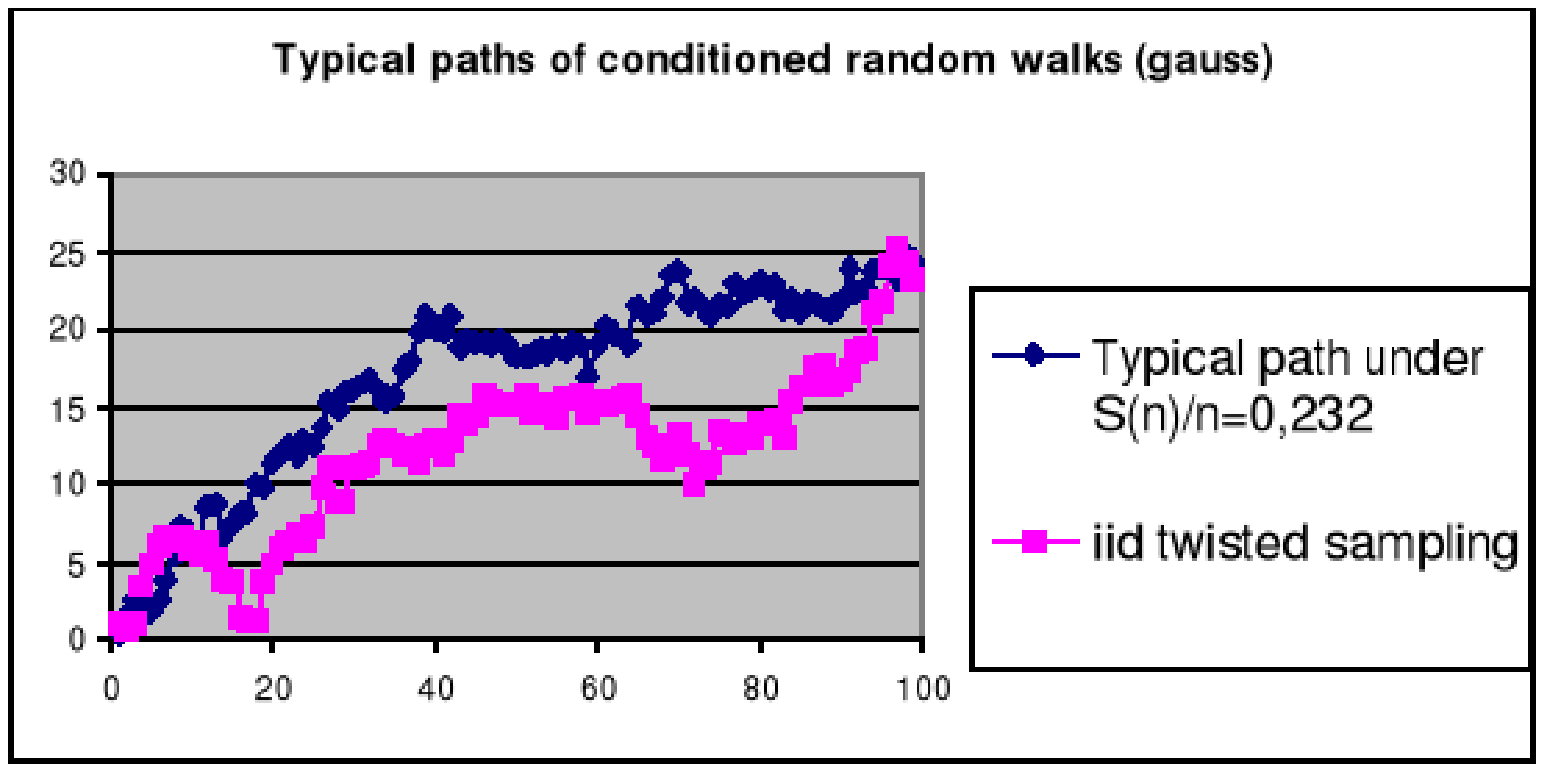}\\
\end{figure}

This graph illustrates Proposition \ref{Prop approx local cond density}.

\subsubsection{Figure 1Gauss}

\begin{figure}
  \centering
  \includegraphics[width=7cm,height=10cm,angle=-90]{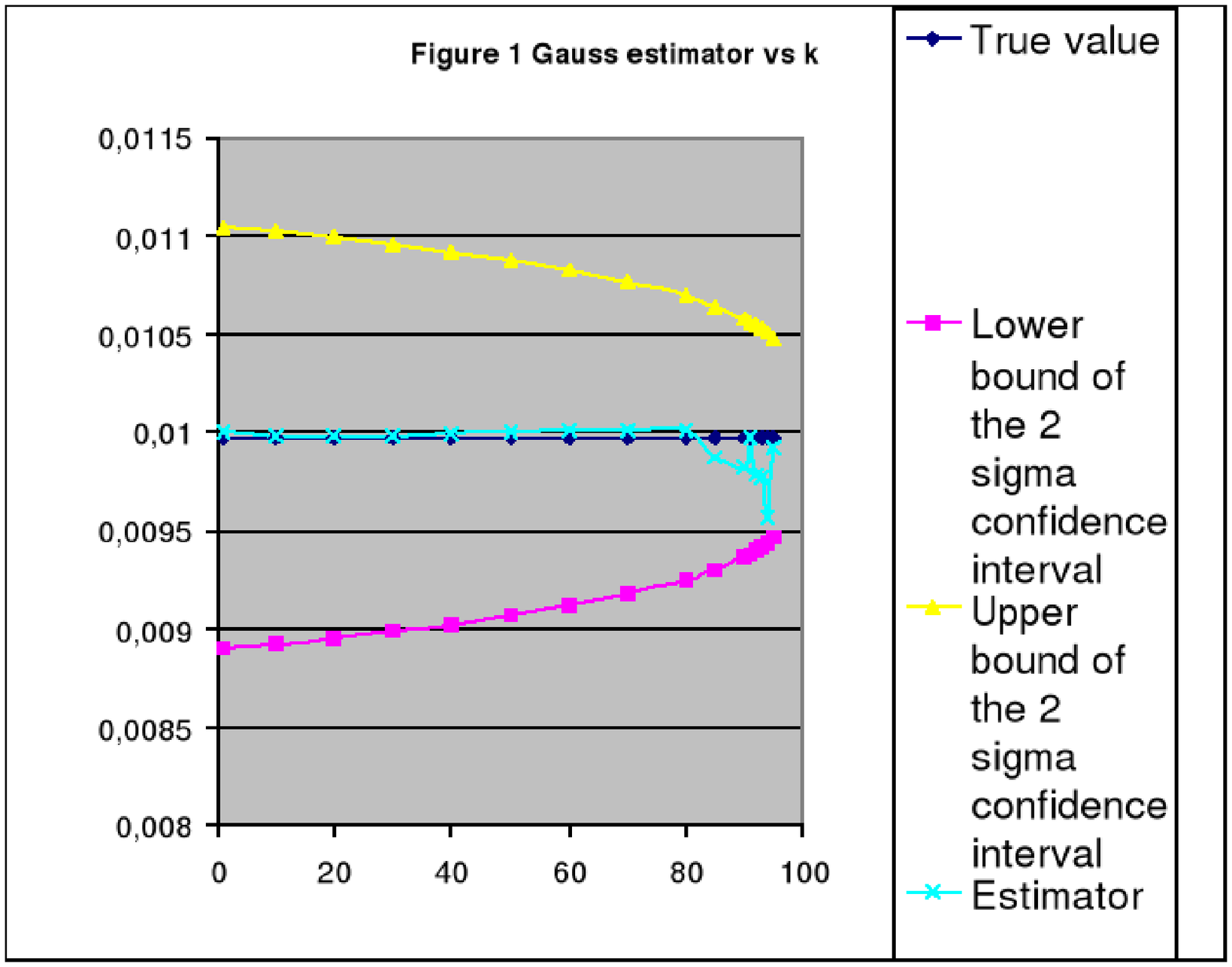}\\
\end{figure}The graph shows the role of $k$ in the behavior of the estimate. The $X_{i}$%
's are standard normal, $n=100$ and $P_{n}=10^{-2}.$ When $k$ is less than $%
70$ the new estimate improves on the classical i.i.d. scheme.\ A change in $M
$ leads no significant change (here $M=30$). The value of $L$ is $L=2000.$%
\bigskip

\bigskip

\subsubsection{Figure 2 Gauss}

\begin{figure}
  \centering
  \includegraphics[width=7cm,height=10cm,angle=-90]{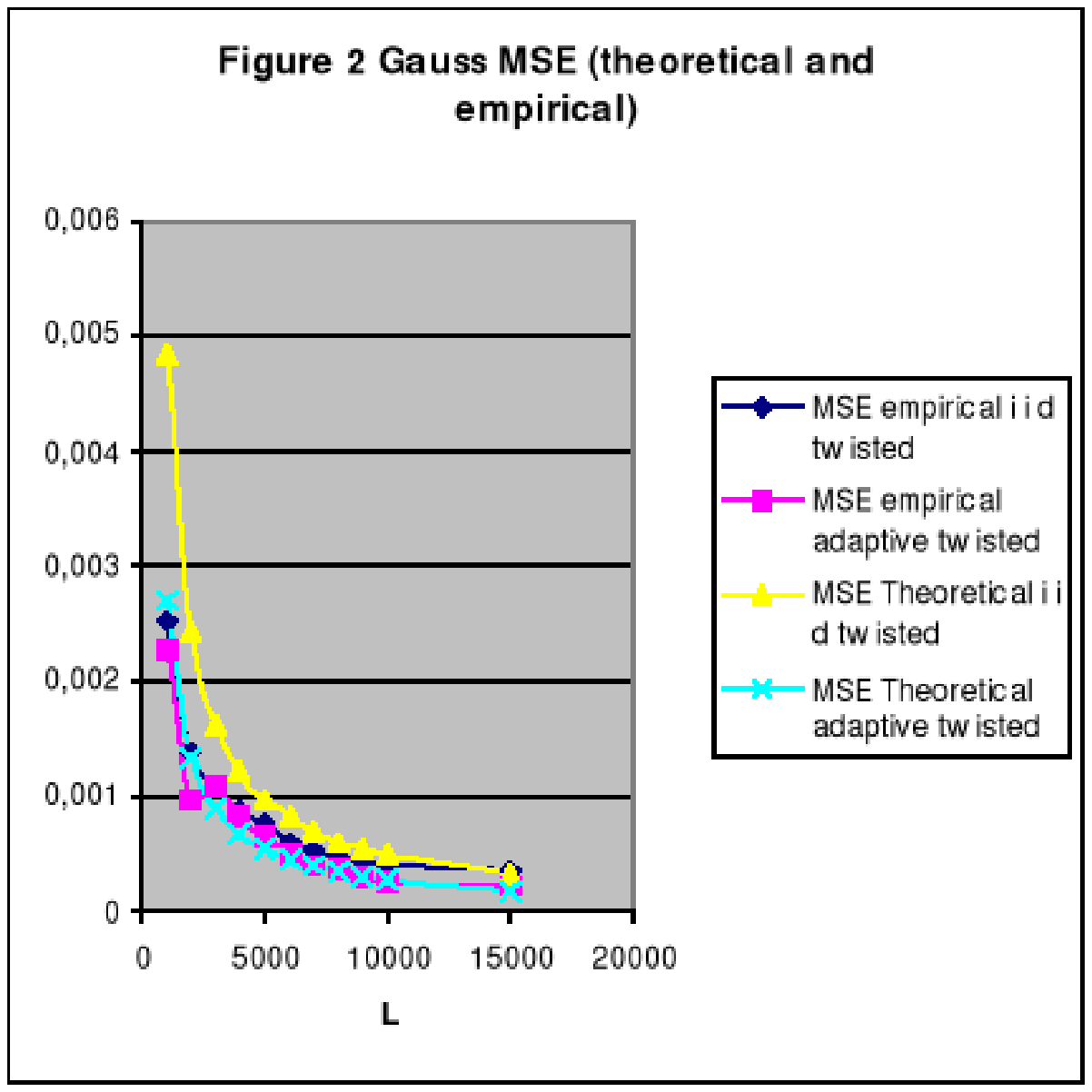}\\
\end{figure}

The graph illustrates the accuracy of the asymptotic results in Propositions %
\ref{Prop rel efficiency standard IS} and \ref{Prop rel efficiency IS}. The $%
X_{i}$'s are standard normal, $n=100,$P$_{n}=10^{-2},k=60.$

\bigskip

\subsubsection{Figure 3 Gauss}

The graph is an illustration of Proposition \ref{Prop rel efficiency}. The r.v's
$X_{i}$'s are standard normal, $n=100$ and $P_{n}=10^{-2}.$ In ordinate is
the ratio of the empirical value of the MSE\ of the adaptive estimate w.r.t.
the empirical MSE of the i.i.d.\ twisted one.\ The value of $k$ is $k=60;$
this ratio stabilizes to $\sqrt{n-k}/\sqrt{n}$ for large $L,$ in full
accordance with Proposition \ref{Prop rel efficiency}.
\begin{figure}
  \centering
  \includegraphics[width=7cm,height=10cm,angle=-90]{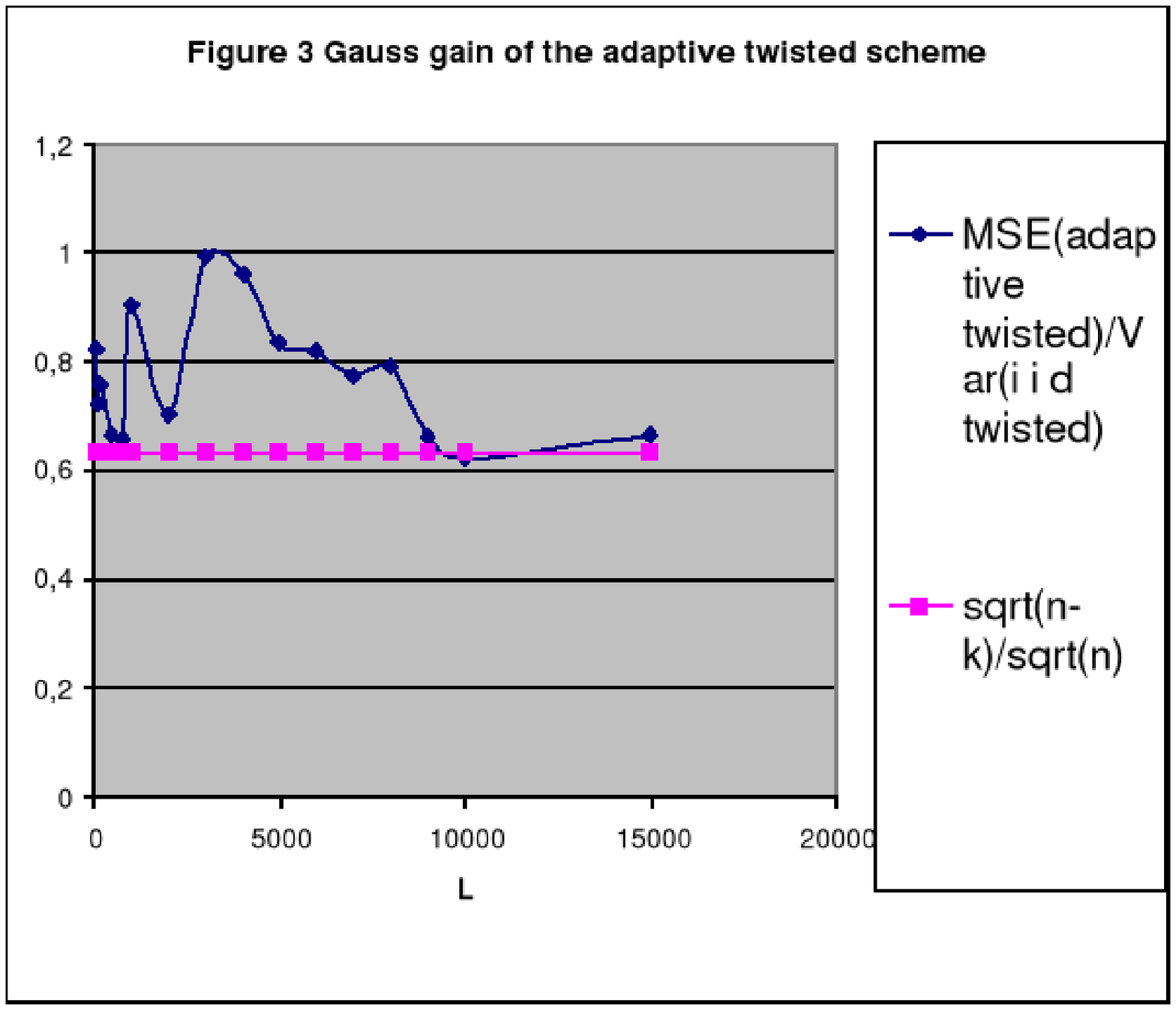}\\
\end{figure}

\subsection{The exponential case}

\subsubsection{typical paths}
\begin{figure}
  \centering
  \includegraphics[width=7cm,angle=-90]{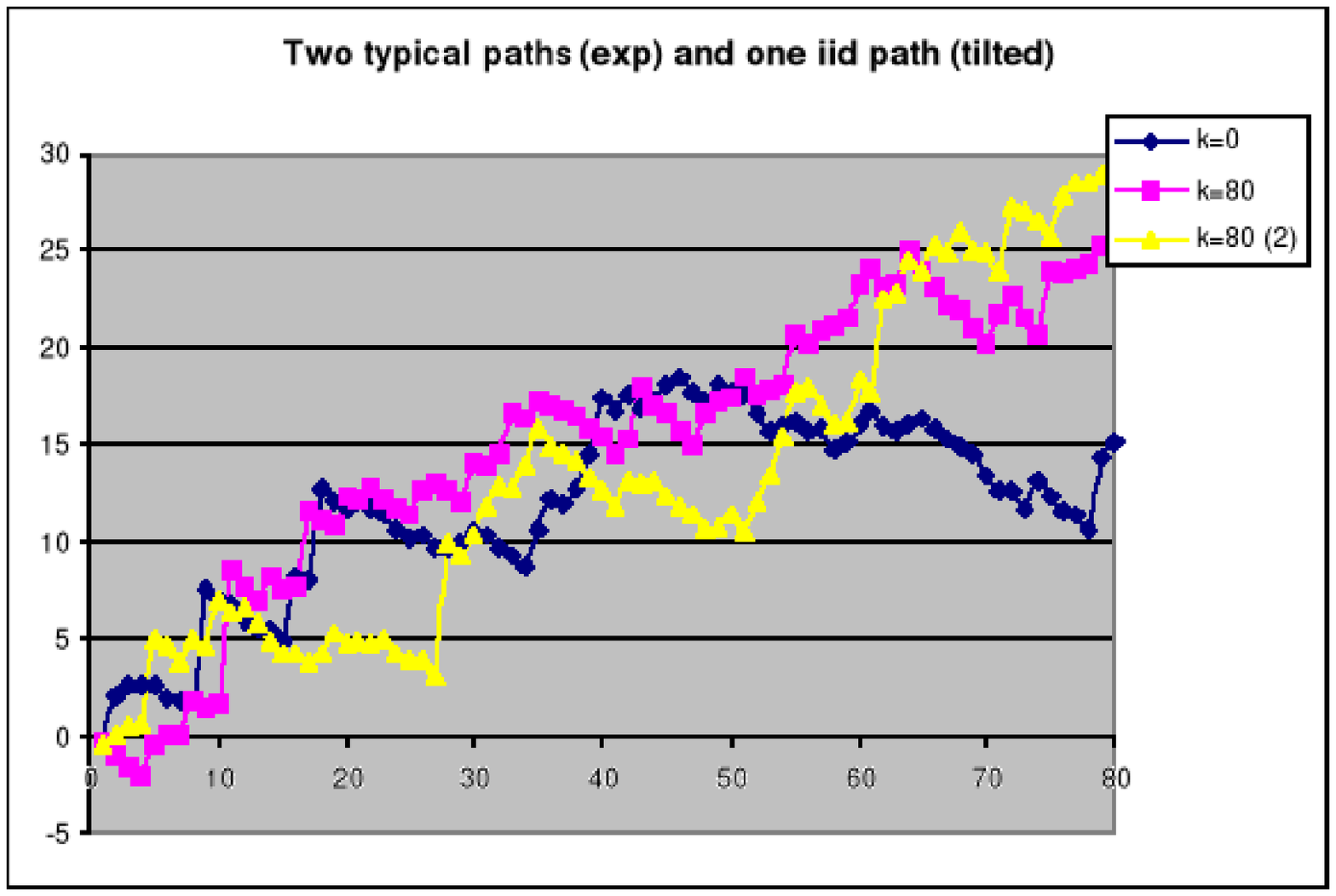}\\
\end{figure}

The graphs above are typical paths under the conditional distribution (with $%
\mathbf{S}_{n}/n=0.239$) and under the i.i.d. sampling with tilted density.
The value of $n$ is $100$ and the approximation of the conditional density
of the random walk is fair up to $k=80$, as indicated by the fact that the
IS estimator of $P_{n}$ is correct up to $k=80$, which can be seen as a
pertinent indicator.

\begin{figure}
  \centering
  \includegraphics[width=6cm,angle=-90]{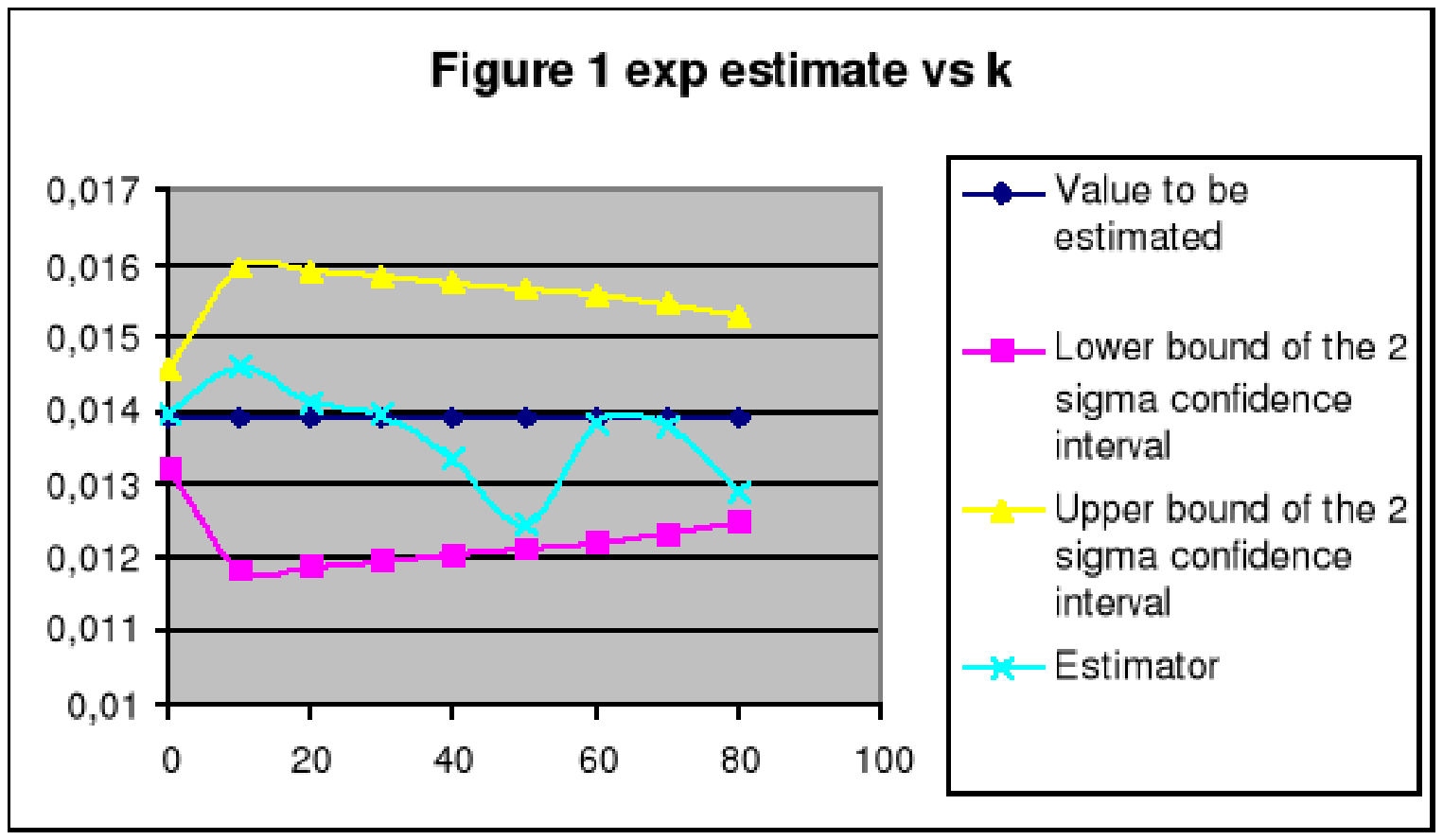}\\
\end{figure}The random variables $%
X_{i}^{\prime }s$ are i.i.d. with exponential distribution with parameter $1$
on $\left( -1,\infty \right) .$ The case treated here is $P\left( \frac{%
\mathbf{S}_{n}}{n}>a_{n}\right) =P_{n}$ with $n=100,$ $P_{n}=0.013887$ and $%
a_{n}=0.232.$ These values are computed through a very long run of the
standard IS algorithm (with i.i.d. sampling according to the twisted) and
are used as a benchmark.The estimates are calculated with $L=1000,$ and $%
L=10000$ for $k=0$, i.e. for the classical i.i.d. twisted sampler (lower
values of $L$ lead unstable estimates)
\section{Appendix}

\subsection{Proof of Proposition \protect\ref{Prop Gibbs Moderate}}

We first state

\begin{lemma}
Denote $q^{\ast }:=\frac{dQ^{\ast }}{d\lambda }$ ($\lambda $ the Lebesgue
measure)$.\;$Then $q^{\ast }(y)=xyp(y)$
\end{lemma}

\begin{proof}
By Theorem 3.4 (2) in Broniatowski and Keziou (2006) it holds $q^{\ast
}(y)=\left( \alpha y+\beta \right) p(y)$ for some constants $\alpha $ and $%
\beta .$ The projection $Q^{\ast \text{ }}$ satisfies both $\int vdQ(v)=x\ $%
\ and $\int dQ(v)=0$ which yield $\alpha =x$ and $\beta =0.$
\end{proof}

For any set $A$ in $B(\mathbb{R}),$ it holds%
\begin{equation}
P\left( \mathbf{X}_{1}\in A/\left( \mathbf{S}_{1}^{n}/n>a_{n}x\right)
\right) =P(A)+a_{n}xQ^{\ast }(A)+o\left( a_{n}\right) .  \label{Gibbs MDP}
\end{equation}

Indeed it holds%
\begin{eqnarray*}
\frac{1}{na_{n}^{2}}\left( P\left( \mathbf{X}_{1}\in A/\mathcal{E}%
_{n,x}\right) -P(A)\right) &=&\frac{1}{na_{n}^{2}}E\left( {\Large 1}_{A}(%
\mathbf{X})-P(A)/\mathcal{E}_{n,x}\right) \\
&=&E\left( \frac{1}{na_{n}^{2}}\left( \frac{1}{n}\sum_{i=1}^{n}{\Large 1}%
_{A}(\mathbf{X}_{i})-P(A)\right) /\mathcal{E}_{n,x}\right) \\
&=&\int_{-\infty }^{0}P\left( \frac{1}{na_{n}^{2}}\left( \frac{1}{n}%
\sum_{i=1}^{n}{\Large 1}_{A}(\mathbf{X}_{i})-P(A)\right) <t/\mathcal{E}%
_{n,x}\right) dt \\
&&+\int_{0}^{\infty }P\left( \frac{1}{na_{n}^{2}}\left( \frac{1}{n}%
\sum_{i=1}^{n}{\Large 1}_{A}(\mathbf{X}_{i})-P(A)\right) >t/\mathcal{E}%
_{n,x}\right) dt.
\end{eqnarray*}

Observe that $\mathcal{E}_{n,x}=\left\{ \mathbf{M}_{n}\in \Omega
_{x}\right\} .$ Also denote $A_{t}^{+}$ (resp $A_{t}^{-}$) the subset of $M(%
\mathbb{R})$ defined through $A_{t}^{+}:=\left\{ Q\in M(\mathbb{R}):Q(%
\mathbb{R})=0,\int {\Large 1}_{A}(v)dQ(v)\geq t\right\} $, resp $%
A_{t}^{-}:=\left\{ Q\in M(\mathbb{R}):Q(\mathbb{R)}=0,\int {\Large 1}%
_{A}(v)dQ(v)<t\right\} .$ Using Bayes formula and the above moderate
deviation result (\ref{Sanov MDP}) it follows that for any measurable set $G$
in $M(\mathbb{R})$

\begin{equation*}
\lim_{n\rightarrow \infty }\Pr \left( M_{n}\in G/M_{n}\in \Omega _{x}\right)
=%
\begin{tabular}{l}
$1$ if $Q^{\ast }$ belongs to $G$ \\
$0$ \ otherwise%
\end{tabular}%
\end{equation*}

\begin{proof}
For any positive (resp. negative) $t$ it then holds $\lim_{n\rightarrow
\infty }\Pr \left( \mathbf{M}_{n}\in A_{t}^{+}/\mathbf{M}_{n}\in \Omega
_{x}\right) =1$ if $t<Q^{\ast }(A)$ and $Q^{\ast }(A)>0$ (resp $%
\lim_{n\rightarrow \infty }\Pr \left( \mathbf{M}_{n}\in A_{t}^{-}/\mathbf{M}%
_{n}\in \Omega _{x}\right) =1$ if $t>Q^{\ast }(A)$ and $Q^{\ast }(A)<0$ ),
which is to say, going to the limit in $n,$ that $\lim_{n\rightarrow \infty }%
\frac{1}{a_{n}}\left( P\left( \mathbf{X}_{1}\in A/\mathcal{E}_{n,x}\right)
-P(A)\right) =\int_{-Q^{\ast -}(A)}^{0}dt+\int_{0}^{Q^{\ast +}(A)}dt$ where $%
Q^{\ast }=Q^{\ast +}-Q^{\ast -}$ is the Lebesgue decomposition of $Q^{\ast
}. $ This closes the proof of (\ref{Gibbs MDP}). A second order expansion of
$\pi ^{a_{n}x}(y)$ in a neighborhood of $t=0$ yields%
\begin{equation*}
\pi ^{a_{n}x}(y)=(1+a_{n}xy+a_{n}^{2}x^{2}g_{n}(y))p(y).
\end{equation*}%
Hence for all Borel set $A$ it holds $\int_{A}\pi
^{a_{n}x}(y)dy=P(A)+a_{n}xQ^{\ast
}(A)+a_{n}^{2}x^{2}\int_{A}g_{n}(y))p(y)dy. $ Since both $\int_{A}\pi
^{a_{n}x}(y)dy$ tends to $P(A)$ and $Q^{\ast }$ is a finite measure it
follows that $a_{n}^{2}x^{2}\int_{A}g_{n}(y))p(y)dy$ tends to $0.$\bigskip
\end{proof}

\subsection{Two Lemmas pertaining to the partial sum under its final value}

We now state two lemmas which describe some functions of the random vector $%
\mathbf{X}_{1}^{n}$ conditioned on $\mathcal{E}_{n}$.

\begin{lemma}\label{Lemmaminunderconditioning}Assume that (A) holds. Then for all $i$
between $1$ and $k$%
\begin{equation*}
\frac{n}{n-i}\left( a_{n}-\frac{\mathbf{S}_{1}^{i}}{n}\right) =a_{n}+O_{%
\mathfrak{P}_{n}}\left( \frac{1}{\sqrt{n-i}}\right) .
\end{equation*}%
\bigskip
\end{lemma}

\begin{proof}
Select $s$ in $(a_{n},b_{n})$ and denote $P_{n}^{s}$ the p.m on $\mathbb{R}%
^{n}$ conditioned on $\left( \mathbf{S}_{1}^{n}=ns\right) $ It holds%
\begin{equation*}
\sqrt{n-i}\left( m_{i,n}-a_{n}\right) =\sqrt{n-i}\left( \frac{\mathbf{S}%
_{i+1}^{n}}{n-i}-s\right) +\sqrt{n-i}\left( a_{n}-s\right) .
\end{equation*}%
We prove that for $m=n-i$
\begin{equation*}
var_{P_{n}^{s}}\left( \sqrt{m}\left( \frac{\mathbf{S}_{1}^{m}}{m}-s\right)
\right) =O(1)
\end{equation*}%
as $m\rightarrow \infty $ where $var_{P_{n}^{s}}Z$ denotes the variance of $%
Z $ conditionally on $\left( \frac{\mathbf{S}_{1}^{n}}{n}=s\right) .$
Integrating with respect to the distribution of $\mathbf{S}_{1}^{n}$
conditioned upon $\mathcal{E}_{n}$ concludes the proof$.$ Using
\end{proof}

\begin{equation*}
p_{n}^{s}(\mathbf{X}_{1}=x)=\frac{p_{\mathbf{S}_{2}^{n}}\left( ns-x\right)
p_{\mathbf{X}_{1}}(x)}{p_{\mathbf{S}_{1}^{n}}\left( ns\right) }=\frac{\pi _{%
\mathbf{S}_{2}^{n}}^{t}\left( ns-x\right) \pi _{\mathbf{X}_{1}}^{t}(x)}{\pi
_{\mathbf{S}_{1}^{n}}^{t}\left( ns\right) }
\end{equation*}%
with $m(t)=s$ , normalizing both $\pi _{\mathbf{S}_{2}^{n}}^{t}\left(
ns-x\right) $ and $\pi _{\mathbf{S}_{1}^{n}}^{t}\left( ns\right) $ and
making use of a first order Edgeworth expansion in those expressions yields
\begin{equation*}
E_{P_{n}^{s}}\left( \mathbf{X}_{1}\right) =s+0\left( \frac{1}{n}\right)
\end{equation*}%
and
\begin{equation*}
E_{P_{n}^{s}}\left( \mathbf{X}_{1}^{2}\right) =s^{2}(t)+s^{2}+0\left( \frac{1%
}{n}\right) .
\end{equation*}%
With a similar development for the joint density $\mathfrak{p}_{n}(\mathbf{X}%
_{1}=x,\mathbf{X}_{2}=y)$, using the same tilted distribution $\pi ^{t}$ it
readily follows that
\begin{equation*}
E_{P_{n}^{s}}\left( \mathbf{X}_{1}\mathbf{X}_{2}\right) =s^{2}+0\left( \frac{%
1}{n}\right) .
\end{equation*}%
Since
\begin{equation*}
var_{P_{n}^{s}}\mathbf{S}_{1}^{m}=m(m-1)E_{P_{n}^{s}}\left( \mathbf{X}_{1}%
\mathbf{X}_{2}\right) +mE_{P_{n}^{s}}\left( \mathbf{X}_{1}^{2}\right)
-m^{2}E_{P_{n}^{s}}\left( \mathbf{X}_{1}\right) ^{2}
\end{equation*}%
it follows that when $m/n$ tends to $0$, then $var_{P_{n}^{s}}\mathbf{S}%
_{1}^{m}=m\left( 1+o(1)\right) .$ Since $m\leq n-k$ this amounts to
\begin{equation*}
\lim_{n\rightarrow \infty }\frac{n-k}{n}=0.
\end{equation*}%
Integration with respect to the distribution of $\mathbf{S}_{1}^{n}$
conditioned upon $\mathcal{E}_{n}$ and splitting the integeral on $%
(a_{n},a_{n}+c_{n})$ and $\left( a_{n}+c_{n},\infty \right) $, using (C2)
concludes the proof.

\begin{remark}
It can be proved that
\begin{equation*}
\sqrt{m}\left( \frac{\mathbf{S}_{1}^{m}}{m}-a_{n}\right) \Rightarrow N(0,1)%
\text{ when }m/n\rightarrow 0
\end{equation*}%
conditionally on $\left( \mathbf{S}_{1}^{n}/n>a_{n}\right) .$ This result is
to be compared with the Gibbs principle for moderate deviations stated in
the Introduction which assets that for \textit{fixed} $m$ the joint
distribution of $(\mathbf{X}_{1},...,\mathbf{X}_{m})$ conditioned upon $%
\mathcal{E}_{n}$ converges weakly , as $n\rightarrow \infty $, to the joint
distribution of $m$ r.v's $\mathbf{X}_{1}^{\ast },...,\mathbf{X}_{m}^{\ast }$
which are independent copies of $\mathbf{X}^{\ast }$ .The above result says
that even for sequences depending upon $n$, we may replace the original $m$
variables by the $m$ \textit{independent} tilted ones when exploring the
behavior of $S_{1}^{m}$ under $\mathcal{E}_{n},$ since $\sqrt{m}\left( \frac{%
\mathbf{S}_{1}^{m}}{m}-a_{n}\right) $ shares the same limit distribution.
\end{remark}

We also need the order of magnitude of $\max \left( \mathbf{X}_{1},...,%
\mathbf{X}_{k}\right) $ under $\mathfrak{P}_{n}$ which is stated in the
following result.

\begin{lemma}\label{Lemma max Y_i under E_n} It holds for all $k$ between $1$ and $n$
\begin{equation*}
\max \left( \mathbf{X}_{1},...,\mathbf{X}_{k}\right) =O_{\mathfrak{P}%
_{n}}(\log n).
\end{equation*}
\end{lemma}

Let $s$ such that $na_{n}\leq s\leq a_{n}$ $+c_{n}$ . Denote $P_{n}^{s}$ the
probability measure of $\mathbf{X}_{1}^{n}$ given the the value of $\mathbf{S%
}_{1}^{n}=s.$ Since
\begin{equation*}
\mathfrak{P}_{n}\left( \max \left( \mathbf{X}_{1},...,\mathbf{X}_{k}\right)
>t\right) =\int_{na_{n}}^{\infty }P_{n}^{s}\left( \max \left( \mathbf{X}%
_{1},...,\mathbf{X}_{k}\right) >t\right) p\left( \mathbf{S}_{1}^{n}=s/%
\mathcal{E}_{n}\right) ds
\end{equation*}%
we first state the order of magnitude of $\max \left( \mathbf{X}_{1},...,%
\mathbf{X}_{k}\right) $ under $P_{n}^{s}$ in the next Lemma.

\begin{lemma}\label{Lemma max X_i under conditioning} For all $k$ between $1$ and $n,\max
\left( \mathbf{X}_{1},...,\mathbf{X}_{k}\right) =O_{P_{n}^{s}}\left( \log
k\right) .$
\end{lemma}

\begin{proof}
Define $\tau :=s/n.$ For all $t$ it holds%
\begin{eqnarray*}
P_{n}^{s}\left( \max \left( \mathbf{X}_{1},...,\mathbf{X}_{k}\right)
>t\right) &\leq &kP_{n}^{s}\left( \mathbf{X}_{n}>t\right) \\
&=&k\int_{t}^{\infty }p(\mathbf{X}_{n}=u/\mathbf{S}_{1}^{n}=s)p(\mathbf{S}%
_{1}^{n}=s/\mathcal{E}_{n})du \\
&=&k\int_{t}^{\infty }\pi ^{\tau }\left( \mathbf{X}_{n}=u\right) \frac{\pi
^{\tau }(\mathbf{S}_{1}^{n-1}=s-u)}{\pi ^{\tau }\left( \mathbf{S}%
_{1}^{n}=s\right) }du.
\end{eqnarray*}%
Center and normalize both $\mathbf{S}_{1}^{n}$ and $\mathbf{S}_{1}^{n-1}$%
with respect to the density $\pi ^{\tau }$ in the last line above, denoting $%
\overline{\pi _{n}^{\tau }}$ the density of $\overline{\mathbf{S}_{1}^{n}}%
:=\left( \mathbf{S}_{1}^{n}-n\tau \right) /s_{\tau }\sqrt{n}$ when $\mathbf{X%
}$ has density $\pi ^{\tau }$ with mean $\tau $ and variance $s_{\tau }^{2},$
we get
\begin{eqnarray*}
P_{n}^{s}\left( \max \left( \mathbf{X}_{1},...,\mathbf{X}_{k}\right)
>t\right) &\leq &k\frac{\sqrt{n}}{\sqrt{n-1}}\int_{t}^{\infty }\pi ^{\tau
}\left( \mathbf{X}_{n}=u\right) \\
&&\frac{\overline{\pi _{n-1}^{\tau }}\left( \overline{\mathbf{S}_{1}^{n-1}}%
=\left( n\tau -u-(n-1)\tau )\right) /\left( s_{\tau }\sqrt{n-1}\right)
\right) }{\overline{\pi _{n}^{\tau }}\left( \overline{\mathbf{S}_{1}^{n}}%
=0\right) }du.
\end{eqnarray*}%
Under the sequence of densities $\pi ^{\tau }$ the triangular array $\left(
\mathbf{X}_{1},...,\mathbf{X}_{n}\right) $ obeys a first order Edgeworth
expansion
\begin{eqnarray*}
P_{n}^{s}\left( \max \left( \mathbf{X}_{1},...,\mathbf{X}_{k}\right)
>t\right) &\leq &k\frac{\sqrt{n}}{\sqrt{n-1}}\int_{t}^{\infty }\pi ^{\tau
}\left( \mathbf{X}_{n}=u\right) \\
&&\frac{\mathfrak{n}\left( \left( \tau -u\right) /s_{\tau }\sqrt{n-1}\right)
\mathbf{P}\left( u,i,n\right) +o(1)}{\mathfrak{n}\left( 0\right) +o(1)}du \\
&\leq &kCst\int_{t}^{\infty }\pi ^{\tau }\left( \mathbf{X}_{n}=u\right) du.
\end{eqnarray*}%
for some constant $Cst$ independent of $n$ and $\tau $ and where
\begin{equation*}
\mathbf{P}\left( u,i,n\right) :=1+P_{3}\left( \left( \tau -u\right) /s_{\tau
}\sqrt{n-1}\right)
\end{equation*}%
where $P_{3}(x)=\frac{\mu _{3}^{(\tau )}}{6\left( \sigma ^{(\tau )}\right)
^{3}}\left( x^{3}-3x\right) $ is the third Hermite polynomial; $\left(
\sigma ^{(\tau )}\right) ^{2}$ and $\mu _{3}^{(\tau )}$ are the second and
third centered moments of $\pi ^{\tau }.$ We used uniformity upon $u$ in the
remaining term of the Edgeworth expansions. Let $t_{\tau }$ such that $%
m(t_{\tau })=\tau .$ Making use of Chernoff Inequality%
\begin{eqnarray*}
P_{n}^{s}\left( \max \left( \mathbf{X}_{1},...,\mathbf{X}_{k}\right)
>t\right) &\leq &\frac{k}{\phi (t_{\tau })}\int_{t}^{\infty }\exp -\left(
C-t_{\tau }\right) u\text{ }du \\
&\leq &kCst\frac{\phi (t_{\tau }+\lambda )}{\phi (t_{\tau })}e^{-\lambda t}
\end{eqnarray*}%
for any $\lambda $ such that $\phi (t_{\tau }+\lambda )$ is finite.
\begin{equation*}
t/\log k\rightarrow \infty
\end{equation*}%
it holds
\begin{equation*}
P_{n}^{s}\left( \max \left( \mathbf{X}_{1},...,\mathbf{X}_{k}\right)
<t\right) \rightarrow 1,
\end{equation*}%
which proves the lemma.
\end{proof}

We now prove Lemma \ref{Lemma max Y_i under E_n}

As above write
\begin{eqnarray*}
\mathfrak{P}_{n}\left( \max \left( \mathbf{X}_{1},...,\mathbf{X}_{k}\right)
>t\right) &\leq &k\mathfrak{P}_{n}\left( \mathbf{X}_{n}>t\right) \\
&\leq &k\int_{na_{n}}^{\infty }\left( \int_{t}^{\infty }\pi ^{\tau }\left(
\mathbf{X}_{n}=u\right) \frac{\pi ^{\tau }(\mathbf{S}_{1}^{n-1}=s-u)}{\pi
^{\tau }\left( \mathbf{S}_{1}^{n}=s\right) }du\right) \\
&&p\left( \mathbf{S}_{1}^{n}=s/\mathcal{E}_{n}\right) ds
\end{eqnarray*}%
where $\tau $ is defined as in the above Lemma through $\tau :=s/n.$ Use the
same argument as in Lemma \ref{Lemma max X_i under conditioning} to assess
that when $t/\log n$ goes to infinity then the.RHS above tends to $0.$ This
closes the proof.

\bigskip

\subsection{Proof of Lemma \protect\ref{Lemma from local cond to global cond}%
}

It holds%
\begin{eqnarray*}
\mathfrak{p}_{n}\left( Y_{1}^{k}\right) &=&\frac{\int_{a_{n}}^{\infty }p_{%
\mathbf{X}_{1}^{k},\mathbf{S}_{n}}\left( Y_{1}^{k},t\right) dt}{P\left(
\mathcal{E}_{n}\right) } \\
&=&\frac{np_{\mathbf{X}_{1}^{k}}(Y_{1}^{k})}{\left( n-k\right) P\left(
\mathcal{E}_{n}\right) }\int_{a_{n}}^{\infty }p_{\mathbf{S}_{k+1}^{n}/\left(
n-k\right) }\left( \frac{n}{\left( n-k\right) }\left( t-\frac{\Sigma _{1}^{k}%
}{n}\right) \right) dt.
\end{eqnarray*}%
By Lemma \ref{Lemma approx exponential for Tbold} it holds under (C1)
\begin{equation*}
\frac{\Sigma _{1}^{n}}{n}=a_{n}+R_{n}
\end{equation*}%
where $R_{n}:=O_{\mathfrak{P}_{n}}\left( \frac{1}{na_{n}}\right) >0.$ Denote
$\mathbf{S:=}\frac{\mathbf{S}_{k+1}^{n}}{n-k}$ $.$ Set
\begin{equation*}
I=\frac{\int_{b}^{\infty }p_{\mathbf{S}}(u)du}{\int_{a}^{\infty }p_{\mathbf{S%
}}(u)du}=\frac{P\left( \mathbf{S}>b\right) }{P\left( \mathbf{S}>a\right) }
\end{equation*}%
with $a:=\frac{n}{n-k}\left( a_{n}-\frac{\Sigma _{1}^{k}}{n}\right) $ and $%
b:=\frac{n}{n-k}\left( b_{n}-\frac{\Sigma _{1}^{k}}{n}\right) .$ It holds%
\begin{equation*}
\mathfrak{p}_{n}\left( Y_{1}^{k}\right) =\left( 1+I\right)
\int_{a_{n}}^{b_{n}}p\left( Y_{1}^{k}/\mathbf{T}=\sigma \right) p\left(
\mathbf{T}=\sigma \right) d\sigma .
\end{equation*}%
Use Lemma \ref{Lemma m_i,n under conditioning} to obtain%
\begin{equation*}
I=\frac{P\left( \mathbf{S}>\alpha _{n}+\frac{n}{n-k}c_{n}\right) }{P\left(
\mathbf{S}>\alpha _{n}\right) }
\end{equation*}%
where $\alpha _{n}:=a_{n}+O_{\mathfrak{P}_{n}}\left( \frac{1}{\sqrt{n-k}}%
\right) =a_{n}\left( 1+o_{\mathfrak{P}_{n}}(1)\right) .$ Use Lemma \ref%
{Lemma Jensen} to obtain
\begin{equation*}
I=\left( \exp -nc_{n}a_{n}\right) \left( \exp \frac{n^{2}c_{n}^{2}}{n-k}%
\right)
\end{equation*}%
which tends to $0$ under (C).\bigskip

\subsection{Proof of Lemma \ref{Lemma set C_n for efficiency}}

The approximation in (\ref{g(X_1^k)}) holds only on
\begin{equation*}
A_{n,\varepsilon _{n}}:=A_{\varepsilon _{n}}^{k}\times \mathbb{R}^{n-k}.
\end{equation*}

In the above display,
\begin{equation*}
A_{\varepsilon _{n}}^{k}:=\left\{ x_{1}^{k}:\left\vert \frac{\mathfrak{p}%
_{n}(x_{1}^{k})}{\mathbf{g}\left( x_{1}^{k}\right) }-1\right\vert
<\varepsilon _{n}\right\} .
\end{equation*}
By the above definition
\begin{equation}
\lim_{n\rightarrow \infty }\mathfrak{P}_{n}\left( A_{n,\varepsilon
_{n}}\right) =1  \label{P_n(a_n,e_n)}
\end{equation}%
Note also that
\begin{eqnarray*}
\mathbf{G}\left( A_{n,\varepsilon _{n}}\right) &:&=\int \mathbf{1}%
_{A_{n,\varepsilon _{n}}}(x_{1}^{n})\mathbf{g}\left( x_{1}^{n}\right)
dx_{1}^{n}=\int \mathbf{1}_{A_{\varepsilon _{n}}^{k}}(x_{1}^{k})\mathbf{g}%
\left( x_{1}^{k}\right) dx_{1}^{n} \\
&\geq &\frac{1}{1+\varepsilon _{n}}\int \mathbf{1}_{A_{\varepsilon
_{n}}^{k}}(x_{1}^{k})\mathfrak{p}_{n}(x_{1}^{k})dx_{1}^{k} \\
&=&\frac{1}{1+\varepsilon _{n}}\left( 1+o(1)\right)
\end{eqnarray*}%
which goes to $1$ as $n$ tends to $\infty ,$ where we have used Proposition %
\ref{Prop p_n equiv g under S_n>na_n}. In the above displays $\mathbf{g}%
\left( x_{1}^{k}\right) $ is the density of $X_{1}^{k}$ when $X_{1}^{n}$ is
sampled under $\mathbf{g}.$ We have just proved that the sequence of sets $%
A_{n,\varepsilon _{n}}$ contains roughly all the sample paths $X_{1}^{n}$
under the importance sampling density $\mathbf{g}.$

We use the fact that $t_{k}$ defined through
\begin{equation*}
m(t_{k})=\frac{n}{n-k}\left( a_{n}-\frac{\Sigma _{1}^{k}}{n}\right)
\end{equation*}%
is close to $a_{n}$ under $\mathfrak{P}_{n}$ uniformly upon $\sigma $ in $%
(a_{n},b_{n}).$

Let $\delta _{n}$ tend to $0$ and $\lim_{n\rightarrow \infty }a_{n}\delta
_{n}\sqrt{n-k}=\infty $ and
\begin{equation*}
B_{n}:=\left\{ x_{1}^{n}:\left\vert \frac{m(t_{k})}{a_{n}}-1\right\vert
<\delta _{n}\right\} .
\end{equation*}

We prove that on $B_{n}$
\begin{equation}
t_{k}s(t_{k})=a_{n}\left( 1+o(1)\right)  \label{t_k s_k}
\end{equation}%
holds.

By Lemma \ref{Lemmaminunderconditioning} and (C3)%
\begin{equation}
\lim_{n\rightarrow \infty }\mathfrak{P}_{n}\left( B_{n}\right) =1.
\label{P_n(B_n)}
\end{equation}

There exists $\delta _{n}^{\prime }$ such that for any $x_{1}^{n}$ in $B_{n}$
\begin{equation}
\left\vert \frac{t_{k}}{a_{n}}-1\right\vert <\delta _{n}^{\prime }.
\label{control t_k/a_n}
\end{equation}%
Indeed
\begin{equation*}
\left\vert \frac{m(t_{k})}{a_{n}}-1\right\vert =\left\vert \frac{t_{k}\left(
1+v_{k}\right) }{a_{n}}-1\right\vert <\delta _{n}
\end{equation*}%
and $\lim_{n\rightarrow \infty }v_{k}=0.$ Therefore%
\begin{equation*}
1-\frac{v_{k}t_{k}}{a_{n}}-\delta _{n}<\frac{t_{k}}{a_{n}}<1-\frac{v_{k}t_{k}%
}{a_{n}}+\delta _{n}.
\end{equation*}%
Since $\frac{m(t_{k})}{a_{n}\text{ }}$ is bounded so is $\frac{t_{k}}{a_{n}}$
and therefore $\frac{v_{k}t_{k}}{a_{n}}\rightarrow 0$ as $n\rightarrow
\infty $ which implies (\ref{control t_k/a_n}).

Further (\ref{control t_k/a_n}) implies that there exists $\delta _{n}"$
such that
\begin{equation*}
\left\vert \frac{t_{k}s(t_{k})}{a_{n}}-1\right\vert <\delta _{n}".
\end{equation*}%
Indeed
\begin{eqnarray*}
\left\vert \frac{t_{k}s(t_{k})}{a_{n}}-1\right\vert &=&\left\vert \frac{%
t_{k}\left( 1+u_{k}\right) }{a_{n}}-1\right\vert \\
&\leq &\delta _{n}^{\prime }+\left( 1+\delta _{n}^{\prime }\right)
u_{k}=\delta _{n}"
\end{eqnarray*}%
where $\lim_{n\rightarrow \infty }u_{k}=0.$ Therefore (\ref{t_k s_k}) holds.

Define
\begin{equation*}
C_{n}:=B_{n}\cap A_{n,\varepsilon _{n}}
\end{equation*}%
Since
\begin{equation*}
\int \mathbf{1}_{C_{n}}(x_{1}^{n})\mathbf{g}\left( x_{1}^{k}\right)
dx_{1}^{n}\geq \frac{1}{1+\varepsilon _{n}}\int \mathbf{1}_{C_{n}}\mathfrak{p%
}_{n}(x_{1}^{n})dx_{1}^{n}
\end{equation*}%
and by (\ref{P_n(a_n,e_n)}) and (\ref{P_n(B_n)})
\begin{equation*}
\lim_{n\rightarrow \infty }\mathfrak{P}_{n}\left( C_{n}\right) =1
\end{equation*}%
we obtain
\begin{equation*}
\lim_{n\rightarrow \infty }\mathbf{G}\left( C_{n}\right) =1.
\end{equation*}%
which concludes the proof.\bigskip

\def\cprime{$'$}


\begin{thebibliography}{10}

\bibitem{Arcones2002}
Miguel~A. Arcones.
\newblock Moderate deviations for {$M$}-estimators.
\newblock {\em Test}, 11(2):465--500, 2002.

\bibitem{BarbeBroniatowski1999}
P.~Barbe and M.~Broniatowski.
\newblock Simulation in exponential families.
\newblock {\em ACM Transactions on Modeling and Computer Simulation (TOMACS)},
  9(3):203--223, 1999.

\bibitem{BroniatowskiKeziou2006}
Michel Broniatowski and Amor Keziou.
\newblock Minimization of {$\phi$}-divergences on sets of signed measures.
\newblock {\em Studia Sci. Math. Hungar.}, 43(4):403--442, 2006.

\bibitem{Bucklew2004}
James~Antonio Bucklew.
\newblock {\em Introduction to rare event simulation}.
\newblock Springer Series in Statistics. Springer-Verlag, New York, 2004.

\bibitem{deAcosta1992}
A.~de~Acosta.
\newblock Moderate deviations and associated {L}aplace approximations for sums
  of independent random vectors.
\newblock {\em Trans. Amer. Math. Soc.}, 329(1):357--375, 1992.

\bibitem{DiaconisFreedman1988}
P.~Diaconis and D.~A. Freedman.
\newblock Conditional limit theorems for exponential families and finite
  versions of de {F}inetti's theorem.
\newblock {\em J. Theoret. Probab.}, 1(4):381--410, 1988.

\bibitem{Ermakov2003}
M.~S. Ermakov.
\newblock Asymptotically efficient statistical inferences for moderate
  deviation probabilities.
\newblock {\em Teor. Veroyatnost. i Primenen.}, 48(4):676--700, 2003.

\bibitem{Ermakov2007}
Mikhail Ermakov.
\newblock Importance sampling for simulations of moderate deviation
  probabilities of statistics.
\newblock {\em Statist. Decisions}, 25(4):265--284, 2007.

\bibitem{Feller1971}
William Feller.
\newblock {\em An introduction to probability theory and its applications.
  {V}ol. {II}.}
\newblock Second edition. John Wiley {\&} Sons Inc., New York, 1971.

\bibitem{FuhWu2004}
Cheng-Der Fuh and Inchi Hu.
\newblock Efficient importance sampling for events of moderate deviations with
  applications.
\newblock {\em Biometrika}, 91(2):471--490, 2004.

\bibitem{InglotKallenbergLedwina1992}
Tadeusz Inglot, Wilbert C.~M. Kallenberg, and Teresa Ledwina.
\newblock Strong moderate deviation theorems.
\newblock {\em Ann. Probab.}, 20(2):987--1003, 1992.

\bibitem{Jensen1995}
Jens~Ledet Jensen.
\newblock {\em Saddlepoint approximations}, volume~16 of {\em Oxford
  Statistical Science Series}.
\newblock The Clarendon Press Oxford University Press, New York, 1995.
\newblock Oxford Science Publications.

\bibitem{Richter1957}
Vol{\cprime}fgang Rihter.
\newblock Local limit theorems for large deviations.
\newblock {\em Dokl. Akad. Nauk SSSR (N.S.)}, 115:53--56, 1957.

\bibitem{Sertfling1980}
Robert~J. Serfling.
\newblock {\em Approximation theorems of mathematical statistics}.
\newblock John Wiley {\&} Sons Inc., New York, 1980.
\newblock Wiley Series in Probability and Mathematical Statistics.

\bibitem{vanCamperhoutCover1981}
Jan~M. Van~Campenhout and Thomas~M. Cover.
\newblock Maximum entropy and conditional probability.
\newblock {\em IEEE Trans. Inform. Theory}, 27(4):483--489, 1981.

\bibitem{Zabell1980}
Sandy~L. Zabell.
\newblock Rates of convergence for conditional expectations.
\newblock {\em Ann. Probab.}, 8(5):928--941, 1980.

\end{thebibliography}

\end{document}